\documentclass[twoside,leqno]{article}
\usepackage[letterpaper]{geometry}
\usepackage{ltexpprt}
\usepackage{hyperref}

\usepackage{lineno}
\usepackage[utf8]{inputenc}
\usepackage{amsfonts}
\usepackage{mathtools}
\usepackage{graphicx}
\usepackage{wrapfig}
\usepackage[capitalize]{cleveref}
\usepackage{makecell}
\usepackage{xspace}
\usepackage{caption}
\usepackage{imakeidx}
\usepackage{thmtools}
\usepackage{todonotes}
\usepackage{enumerate}
\usepackage[shortlabels]{enumitem}
\usepackage{thm-restate}

\date{}
\newcommand{\eps}{\varepsilon}
\newcommand{\Order}{O}
\newcommand{\Oh}{O}
\newcommand{\frechet}{Fr\'{e}chet }

\newcommand{\DTW}{\textup{DTW}}
\newcommand{\FD}{FD}

\newcommand{\oEight}{\star^8}

\newcommand{\Real}{\mathbb{R}}

\newcommand{\orc}{\includegraphics[height=\fontcharht\font`A]{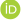}}

\newtheorem{definition}{Definition}[section]

\newtheorem{assumption}{Assumption}[section]

\title{\Large Dynamic Dynamic Time Warping}
\author{ Karl Bringmann \thanks{Saarland University, Saarland Informatics Campus, Saarbrucken, Germany} \thanks{Max Planck Institute for Informatics, Saarland Informatics Campus, Saarbrücken, Germany}
\and Nick Fischer \thanks{The Weizmann Institute of Science, Department of Computer Science and Applied Mathematics, Rehovot, Israel}
\and Ivor van der Hoog \href{https://orcid.org/0009-0006-2624-0231}{\orc} \thanks{Technical University of Denmark, Kongens Lyngby, Denmark}
\and Evangelos Kipouridis \href{https://orcid.org/0000-0002-5830-5830}{\orc} \footnotemark[1] \footnotemark[2]
\and Tomasz Kociumaka \href{https://orcid.org/0000-0002-2477-1702}{\orc} \footnotemark[2]
\and Eva Rotenberg \href{https://orcid.org/0000-0001-5853-7909}{\orc} \footnotemark[4]
}

\newcommand{\floor}[1]{\lfloor #1 \rfloor}
\newcommand{\ceil}[1]{\lceil #1 \rceil}
\newcommand{\set}[1]{\{ #1 \}}
\newcommand{\Intermediary}{\textsc{Intermediary}}

\newcommand{\gadgetSize}{g}
\newcommand{\blockFormula}{2\tau_1+2}
\newcommand{\block}{B}
\newcommand{\mainRowFormula}{(\tau_2-1)\block{}+\tau_1}
\newcommand{\mainRow}{Y_1}
\newcommand{\mainHeightFormula}{2\tau_1+\tau_3}
\newcommand{\mainHeight}{H}
\newcommand{\mainRowEndFormula}{\mainRow{} + \mainHeight{}}
\newcommand{\mainRowEnd}{Y_2}
\newcommand{\cols}{\tau_2\block{}+1} 
\newcommand{\rows}{2Y_1+H} 
\newcommand{\consts}{n_1}
\newcommand{\binaryNum}{n_2}
\newcommand{\nodes}{n_3}
\newcommand{\smallCliqueGadgetA}{n_4}
\newcommand{\smallCliqueGadgetB}{n_5}
\newcommand{\cliqueGadget}{n_6}

\newcommand{\constGadget}{A_4}
\newcommand{\belowConstGadget}{A_3}
\newcommand{\constEnabler}{A_2}
\newcommand{\constStayBigDiagonal}{A_1}
\newcommand{\constStaySmallDiagonal}{A_0}
\newcommand{\constWeight}{W}

\newtheorem{observation}{Observation}[section]

\begin{document}

\maketitle


\begin{abstract}
    \noindent
The Dynamic Time Warping (DTW) distance is a popular similarity measure for polygonal curves (i.e., sequences of points). It finds many theoretical and practical applications, especially for temporal data, and is known to be a robust, outlier-insensitive alternative to the \frechet distance. For static curves of at most $n$ points, the DTW distance can be computed in $O(n^2)$ time in constant dimension. This tightly matches a SETH-based lower bound, even for curves in $\mathbb{R}^1$.

In this work, we study \emph{dynamic} algorithms for the DTW distance. Here, the goal is to design a data structure that can be efficiently updated to accommodate local changes to one or both curves, such as inserting or deleting vertices and, after each operation, reports the updated DTW distance. We give such a data structure with update and query time $O(n^{1.5} \log n)$, where $n$ is the maximum length of the curves.

As our main result, we prove that our data structure is conditionally \emph{optimal}, up to subpolynomial factors. More precisely, we prove that, already for curves in $\mathbb{R}^1$, there is no dynamic algorithm to maintain the DTW distance with update and query time~\makebox{$O(n^{1.5 - \delta})$} for any constant $\delta > 0$, unless the Negative-$k$-Clique Hypothesis fails.
In fact, we give matching upper and lower bounds for various trade-offs between update and query time, even in cases where the lengths of the curves differ.

\end{abstract}

\paragraph{Acknowledgements.}
This work is part of the project CONJEXITY that has received funding
from the European Research Council (ERC) under the European Union’s
Horizon Europe research and innovation programme (grant agreement No.
101078482).
This work is part of the project TIPEA that has received funding from the European Research Council (ERC)
under the European Union's Horizon 2020 research and innovation programme (grant agreement No. 850979).
This research was supported by Independent Research Fund Denmark grant 2020-2023 (9131-00044B) ``Dynamic Network Analysis'' and the VILLUM Foundation grant 37507 ``Efficient Recomputations for Changeful Problems''.
This project has additionally received funding from the European Union's Horizon 2020 research and innovation programme under the Marie Sk\l{}odowska-Curie grant agreement No 899987.
This research benefited from meetings at the Max Planck Institute for Informatics, Saarbrücken, and from discussions at Dagstuhl Seminar 22461 `Dynamic Graph Algorithms'.

\newpage

\section{Introduction} \label{sec:intro}
Sequence similarity measures are fundamental in computational geometry and string algorithms and serve as essential tools in a variety of application domains for analyzing time-ordered data, including videos, audio files, time-series measurements, and GPS tracking data. One of the most widely used similarity measures is the  \emph{Dynamic Time Warping} (DTW) distance, with applications in speech recognition~\cite{myers1980performance,sakoe1978dynamic}, handwriting and online signature matching~\cite{efrat2007curve, tappert1990state, munich1999continuous}, gesture recognition~\cite{corradini2001dynamic, kuzmanic2007hand}, medicine~\cite{caiani1998warped,aach2001aligning}, song recognition~\cite{muller2006efficient, zhu2003warping}, motion retrieval \cite{muller2007dtw}, time series clustering \cite{niennattrakul2007clustering}, and time series database search~\cite{gu2006simple, kahveci2001variable}. See also the survey by Senin~\cite{senin2008dynamic}.

In this paper, we pick up the study of the well-motivated \emph{dynamic} Dynamic Time Warping problem, where the goal is maintaining the DTW distance of two sequences that undergo changes. We provide (1) a new dynamic algorithm that is significantly faster than the previous state of the art, and (2) strong evidence that this algorithm is  \emph{optimal}, up to lower-order factors, conditioned on a well-established hardness assumption from fine-grained complexity theory. We thus conditionally resolve the time complexity of dynamic DTW, successfully closing the problem.

\paragraph{Dynamic Time Warping}
The DTW distance of two sequences (or \emph{curves}) $P = (p_1,\ldots,p_n)$ and $Q = (q_1,\ldots,q_m)$ is defined as follows. Imagine a dog walking along $P$ and its owner walking along $Q$. Both owner and dog start at the beginning of their curves, and in each step the owner may stay in place or jump to the next point along $P$ and the dog may stay in place or jump to the next vertex along $Q$, until both of them have reached the end of their curves. Formally, this yields a traversal $T = ((i_1,j_1), \ldots, (i_t,j_t))$ with~\makebox{$i_1 = j_1 = 1$}, $i_t = n$, $j_t = m$ and $(i_{k+1},j_{k+1}) \in \{(i_k + 1, j_k), (i_k, j_k + 1), (i_k + 1, j_k + 1)\}$ for each step $k \in \{1,\ldots,t-1\}$. The cost of a traversal is the sum of all distances of dog and owner during the traversal, that is, the cost of traversal $T$ is $\sum_{k=1}^t d(p_{i_k}, q_{j_k})$. 
The DTW distance of $P$ and $Q$, denoted $\DTW(P, Q)$, is then defined as the minimum cost of any traversal.

Note that for this definition to make sense, we need to fix a distance measure $d(\cdot,\cdot)$ on vertices. In a typical scenario, the vertices $p_1,\ldots,p_n, q_1,\ldots, q_m$ lie in a low-dimensional Euclidean space~$\mathbb{R}^{d'}$ and their distance is an $L_p$ norm $d(x,y) = \|x-y\|_p$. Throughout this paper, for the algorithms we assume that $d(\cdot,\cdot)$ can be evaluated in constant time (apart from that, it can be arbitrary), and the lower bounds apply to curves in $\mathbb{R}^1$ (meaning that vertices lie in $\mathbb{R}^1$ and the distance measure is $d(x,y) = |x-y|$). This means that both algorithms and lower bounds apply very generally to DTW for various distance measures $d(\cdot, \cdot)$, for example, to $\mathbb{R}^{O(1)}$ under any $L_p$ norm. 

\paragraph{Static DTW Algorithms}
The time complexity of computing the DTW distance of two static curves is well-understood: Given two curves $P$ and $Q$ of length $n$ and $m$, we can compute their DTW distance in time $O(nm)$ by a simple dynamic programming algorithm; the time becomes~$O(n^2)$ when both curves have length at most $n$. This running time is almost the best known, up to mild improvements~\cite{GoldS18}. Moreover, assuming the Strong Exponential Time Hypothesis (or the Orthogonal Vectors Hypothesis) from fine-grained complexity theory, there is a \emph{matching} lower bound stating that DTW cannot be computed in truly subquadratic time $O(n^{2-\delta})$, for any constant $\delta > 0$~\cite{abboud2015tight,bringmann2015quadratic}. This lower bound applies even for curves in~$\mathbb{R}^1$.
Even if we relax the goal to a constant-factor approximation, no algorithm running in truly subquadratic time is known (see \cite{kuszmaul2019dynamic,AgarwalFPY16,YingPFA16} for polynomial-factor approximation algorithms and approximation algorithms for restricted input models), though for this approximate setting conditional lower bounds are still amiss.

\paragraph{Dynamic DTW Algorithms}
In this paper we study DTW on dynamically changing sequences. This problem was introduced by Nishi et al.~\cite{NNIBT20}. Here, both curves receive \emph{updates} that insert a vertex anywhere in the sequence, or remove any vertex from the sequence, or substitute any vertex by a new vertex. We may also receive \emph{queries}, upon which we should return the DTW distance of $P$ and $Q$. Our goal is to store $P$ and $Q$ in a data structure that supports these updates and queries. 
In short: The sequences $P$ and $Q$ undergo updates such as insertion, deletion, or substitutions of a vertex, and for each query we want to recompute their current DTW distance. 

The motivation for this dynamic problem is threefold \cite{NNIBT20}: (1) Big data applications produce quickly changing data, and the time-ordered data relevant for DTW is no exception. (2) For audio or video applications one can imagine that a file gets edited and after each edit some similarity score should be recomputed. (3) The time complexity of dynamic versions of sequence similarity measures is of inherent algorithmic interest; see the related work section below for further examples.

Note that the results for static algorithms yield some upper and lower bounds for the time complexity of the dynamic DTW problem. We focus for simplicity on the case that update and query time must be equal. Then by using a static algorithm to recompute the DTW distance we get update/query time $O(n^2)$, if both curves have length at most $n$. 
Nishi et al.~\cite{NNIBT20} improved this to update/query time $O(n + m + \#chg)$, where $\#chg$ is the number of changes in the dynamic programming table. Since in the worst case $\#chg = \Theta(n m)$, their algorithm yields the same update/query time $O(n^2)$ as the static recomputation, in case both curves have length about $n$.
On the other hand, starting from two empty curves and using $O(n)$ updates we can create any worst-case instance of size $n$, and thus the static lower bound implies that $O(n)$ updates and one query cannot be done in time $O(n^{2-\delta})$ for any $\delta>0$ assuming the Strong Exponential Time Hypothesis. It follows that the update/query time cannot be $O(n^{1-\delta})$. However, these static bounds leave a large gap.

\subsection{Our Results}
In this paper, we improve upon both the quadratic upper bound and the linear lower bound. Specifically, for the upper bound we design a new algorithm for dynamic DTW distance with update and query time $O(n^{1.5} \log n)$. For the lower bound, we show that the update and query time cannot be improved to $O(n^{1.5-\delta})$ for any constant $\delta > 0$ assuming the Negative-$k$-Clique Hypothesis.
Since our upper and lower bound match, we fully resolve the time complexity of the dynamic DTW distance, up to lower order factors and assuming the Negative-$k$-Clique Hypothesis. 

In fact, our precise results do not assume that query and update time are equal. Instead, we obtain trade-off results both for the upper and lower bound, and both trade-offs match up to lower order factors.
In the remainder of this section we discuss these results in detail. 

Note that a substitution update can be simulated by a deletion and an insertion. Therefore, any data structure that supports insertions and deletions automatically also supports substitutions. 

\paragraph{Upper Bound}
For a dynamic DTW data structure on curves $P,Q$, we denote by $n$ and $m$ the current length of $P$ and $Q$ respectively at any point in the update and query sequence. To simplify notation, we assume that $n \ge m$ always holds.\footnote{This assumption can be easily removed, but then in the time bounds $n$ is replaced by $\max\{n,m\}$ and $m$ is replaced by $\min\{n,m\}$.}
In Section~\ref{sec:upper-bound} we design a data structure for dynamic DTW with the following guarantees.

\begin{restatable}{theorem}{tradeoff}
\label{thm:upperboubnd} 
  For any constant $\beta \in [0, \frac{1}{2}]$, there is a data structure that maintains curves $P$ and~$Q$ of (changing) lengths $n$ and~$m$ and supports insertion and deletion updates in time $O(n m^\beta \log m)$ and DTW queries in time $O(n m^{1-\beta} \log m)$. The data structure takes space $O(n m \log m)$. 
\end{restatable}

We give a high-level description of this data structure.
The curves $P$ and $Q$ induce a vertex-weighted grid graph $G$ of $n$ columns and $m$ rows, where we have horizontal, vertical, and diagonal edges.
We call such a graph a \emph{rectangular} graph, referencing its rectangular bounding box. 
The value $\DTW(P, Q)$ corresponds to the length of the shortest $xy$-monotone path\footnote{A sequence of pairs $(x_1,y_1), (x_2,y_2), \ldots, (x_\ell,y_\ell)$ such that both $x_1,x_2,\ldots, x_\ell$ and $y_1,y_2,\ldots, y_\ell$ are monotone.} from node $(1, 1)$ to node $(n, m)$. 
Throughout this paper, our coordinates refer to matrix indexing. Hence, we consider the shortest $xy$-monotone path from the top left to the bottom right corner. 
We show an upper bound that is more general than Theorem~\ref{thm:upperboubnd}. We choose some $\beta \in [0, \frac{1}{2}]$. We partition $P$ and $Q$ into $\Theta( (n + m) m^{-\beta})$ subcurves containing $O(m^{\beta})$ vertices each. 
This partitions the rectangular graph into $O(n m^{1 - 2 \beta} )$ rectangular subgraphs, each containing $O(m^{2 \beta})$ vertices. 
For each rectangular subgraph~$G'$, we store a distance matrix $A'$ with dimensions $m^{\beta} \times m^{\beta}$.
This matrix considers all vertices $S$ on the top and left boundaries of the bounding rectangle and all vertices $T$ on the bottom and right boundaries, and records (by using known techniques for multiple-source shortest paths in planar graphs) all $st$-distances for $s \in S$ and $t \in T$. 
At query time, we combine the $O( n  m^{1 - 2 \beta} )$ matrices, spending $O(m^{\beta}\log m)$ time per matrix, by a wavefront algorithm to compute $\DTW(P, Q)$ in $O(n m^{1 - \beta} \log m)$ time.  
At update time, intuitively, any update affects $O(nm^{-\beta})$ matrices $A'$ (assuming $n \geq m$). We can update each matrix $A'$ in time $O(m^{2\beta} \log m)$, and thus we can perform the whole update in time $O( n m^{- \beta} \cdot m^{2\beta} \log m) = O( n m^{\beta} \log m)$.

\paragraph{Lower Bound}
We show that the above trade-off of update and query time is optimal, up to subpolynomial factors and assuming a plausible hypothesis from fine-grained complexity theory.
The specific hypothesis concerns the Negative-$k$-Clique problem for some constant $k$: Given an undirected $k$-partite graph $G$ with $N$ nodes and integer edge-weights of absolute value at most~$N^{O(k)}$, decide whether $G$ has a $k$-clique with negative total edge weight. This problem has a naive~$O(N^k)$-time algorithm, and no much faster algorithms are known. This lack of progress lead to the formulation of the \emph{Negative-$k$-Clique Hypothesis}, which postulates that Negative-$k$-Clique cannot be solved in time $O(N^{k-\delta})$ for any constants $k \ge 3$ and $\delta > 0$. For $k=3$ this hypothesis is equivalent to the famous All-Pairs Shortest Paths hypothesis~\cite{WilliamsW10}, so it is very believable.
The generalization to $k > 3$ is very plausible, and was successfully applied as a hardness assumption in several contexts~\cite{AbboudBDN18,AbboudWW14,BackursDT16, BackursT17,BringmannCM22,BringmannGMW20, LincolnWW18}, even as the basis for public-key cryptography schemes~\cite{LaVigneLW19}. Moreover, the Negative-$k$-Clique Hypothesis is known to imply the Orthogonal Vectors Hypothesis~\cite{AbboudBDN18}, which implies the static lower bounds for DTW.
We use the Negative-$k$-Clique Hypothesis to prove the following matching lower bound for our dynamic DTW algorithm.

\begin{restatable}{theorem}{lowerboundDTW}
  Fix constants $c\ge 1, \beta \in [0,\frac12]$, $\gamma \in (0, 1]$ and $\delta, \varepsilon > 0$.
  If there is a data structure that maintains curves $P$ and $Q$ of length $n$ and $m$ with $m \in \Omega(n^{\gamma - \eps}) \cap O(n^{\gamma + \eps})$, has preprocessing time $O((n+m)^c)$, and supports substitution updates in time $\Oh(n \cdot m^{\beta-\delta})$ and DTW queries in time $\Oh(n \cdot m^{1-\beta-\delta})$,
    then the Negative-$k$-Clique Hypothesis fails (for some constant $k$ depending solely on $(c, \beta, \gamma, \delta, \varepsilon)$).
\end{restatable}

We next give a high-level overview of this conditional lower bound.
We first define a new intermediary algorithmic problem which we call \textsc{Intermediary}.
Then we present reductions from Negative-$k$-Clique to \textsc{Intermediary}, and from \textsc{Intermediary} to dynamic DTW. 

\begin{figure}[t]
  \centering
  \includegraphics[]{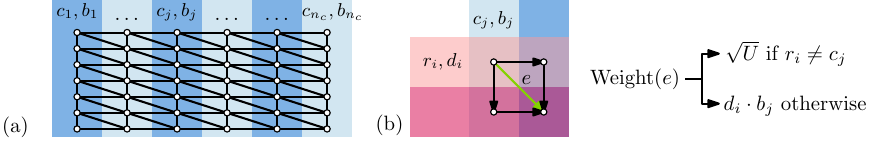}
  \caption{
    We illustrate the intermediary problem for our lower bound. 
  (a) A rectangular graph with $n_r$ rows and $n_c$ columns.
(b) Vertical and horizontal edges always have weight $U$. 
The diagonal connecting $(i, j)$ and $(i+1, j+1)$ has weight $\sqrt{U}$ if $r_i \neq c_j$, and it has weight $d_i b_j$ otherwise. 
    }
 \label{fig:intermediary}
\end{figure}

\begin{restatable}[See Figure~\ref{fig:intermediary} for an illustration]{problem}{intermediary}
In the \textsc{Intermediary} problem, the task is to maintain a directed $(n_r\times n_c)$-grid graph, 
with the weights of horizontal, vertical, and diagonal edges defined using auxiliary parameters:
\begin{itemize}[noitemsep, nolistsep]
    \item non-negative integer \emph{identifiers} $r_i$ and $c_j$ for all rows $i\in [n_r]$ and columns $j\in [n_c]$, resp.;
    \item non-negative integer \emph{weights} $d_i$ for all rows $i\in [n_r]$;
    \item Booleans $b_j$ for all columns $j\in [n_c]$,
    \item an integer $U > n_r \cdot n_c \cdot (\max_i d_i  \cdot \max _i r_i  \cdot \max_j c_j)^2$.
\end{itemize}
Horizontal edges $(i, j)\to (i, j+1)$ and vertical edges $(i, j)\to (i+1, j)$ have weight $U$. 
The weight of the diagonal edge $(i, j)\to (i+1, j+1)$ is $d_i \cdot b_j$ if $r_i = c_j$, and $\sqrt{U}$ otherwise. \newline
An \emph{update} $\textsc{update}(j,x)$ sets $b_j:=x$ (for $x \in \{0, 1\}$). \newline
A \emph{query} returns the length of the shortest $(0, 0)\leadsto (n_r-1, n_c-1)$ path, unless the length is at least $U \cdot |n_r - n_c| + \sqrt{U}$, in which case it returns $\infty$.  
\end{restatable}

In Section~\ref{sec:lower-bound-curve} we reduce from \textsc{Intermediary} to dynamic DTW:

\begin{restatable}{theorem}{curvelowerbound}
    Suppose that there is a data structure that maintains curves $P,Q$ (subject to substitutions of vertices in $Q$ and DTW queries) with preprocessing time $T_P(n,m)$, update time $T_U(n,m)$, and query time $T_Q(n,m)$, where $n=|P|$ and $m=|Q|$.
    Then any $(n_r \times n_c)$-size instance of the \textsc{Intermediary} problem can be solved with preprocessing time $\Oh(T_P(n,m)+n+m)$, update time $\Oh(T_U(n,m))$, and query time $\Oh(T_Q(n,m))$, where $n = \Oh(n_r)$ and $m=\Oh(n_c)$.
\end{restatable}


In Section~\ref{sec:lowerbound} we then reduce from Negative-$k$-Clique to \textsc{Intermediary}:


\begin{restatable}{theorem}{lowerBoundIntermediary} \label{thm:lowerboundInter}
    Fix constants $c\ge 1$, $\beta \in [0,\frac12]$, $\gamma \in (0, 1]$, and $\delta, \varepsilon > 0$.
    If the \textsc{Intermediary} problem for $n_c \in \Omega(n_r^{\gamma - \varepsilon}) \cap O(n_r^{\gamma + \varepsilon})$ can be solved with preprocessing time $\Oh((n_r + n_c)^{c})$,
    update time $\Oh(n_r \cdot n_c^{\beta-\delta})$, and query time $\Oh(n_r \cdot n_c^{1-\beta-\delta})$,
    then the Negative-$k$-Clique Hypothesis fails (for some constant $k$ depending solely on $(c, \beta, \gamma, \delta, \varepsilon)$).
\end{restatable}

Notice that \Intermediary{} is in fact a dynamic shortest path problem on planar graphs.
In \cite{abboudD16} a lower bound is given for a dynamic shortest path problem on planar graphs, where updates select a single edge and modify its weight.
However, in \Intermediary{} we are not allowed to modify the weight of single edges.
In fact the weight of an edge can take at most two different values.

In our lower bound, we design gadgets that consist of batches of edges.
We then ensure that a shortest path either follows all the edges in a batch, or none of them.
By properly modifying the weight of each edge in a batch, we simulate the capability of choosing between multiple weights, instead of just two, for each gadget.

This is again not enough, because an update changes the weights of multiple edges (and also gadgets), not just a selected one.
Let us now provide some (oversimplified) intuition on how we deal with this difficulty.
All gadgets can be thought of as having two different weights, a large and a small one.
The large weights are never modified, and they ensure a restricted form for the shortest path.
The small weights can be thought of as a noise on top of the large weights, that cannot invalidate the aforementioned shortest path property, but encodes the critical information for the lower bound.
Now, when we modify the weight of a gadget, this may indeed ``accidentally'' modify the (small) weight of some other gadget as well.
However, it is ensured that when the weight of a gadget is accidentally changed, then this gadget could anyway not be part of a shortest path, due to its large weight.

The actual construction is of course more elaborate, and at times needs to modify the large weights as well.

\subsection{Further Related Work}

\paragraph{Static Dynamic Time Warping}
DTW is a well-studied similarity measure with a wealth of prior work. 
Recall that in the static setting computing DTW requires essentially quadratic time. Attempting to bypass this barrier, previous works studied DTW from many creative angles, e.g., on binary inputs (i.e., both curves have at most two different vertices)~\cite{abboud2015tight,K21}, on run-length encoded curves~\cite{SDHDKJ18,XK22,FroeseJRW23,BGMW23}, parametrized by the DTW distance~\cite{kuszmaul2019dynamic}, approximation algorithms~\cite{AgarwalFPY16,YingPFA16,kuszmaul2019dynamic}, communication complexity~\cite{BCKY20}, and many other settings~\cite{GDPS22,SI20,SI22,HG19,HG22,BringmannKKMN22}. 
DTW is also related to other similarity measures in computational geometry and stringology:

\paragraph{\frechet Distance}
The \frechet distance is defined similarly as DTW, except that the cost of a traversal is the maximum distance that dog and owner have at any point during the traversal (instead of the sum of all distances as in DTW). 
\frechet distance has the nice feature of being a metric, but it is less outlier resistant compared to the DTW distance. 
\frechet distance can also be computed in time $O(n^2)$~\cite{eitermannila94,buchin2017four}, and not in time $O(n^{2-\delta})$ for any $\delta > 0$ assuming the Strong Exponential Time Hypothesis~\cite{bringmann2014walking}. Similarly as DTW, also \frechet distance has been studied from various angles, see, e.g.~\cite{BringmannM16,ChanR18,FiltserK20,BringmannKN21,HorstKOS23}. 
To the best of our knowledge, \frechet distance has not been studied in a dynamic setting. 

The line of research that comes closest to a dynamic setting are nearest neighbor data structures for the \frechet distance. Here, a set $S$ of curves is given as a static input that we can preprocess to build a data structure. As a query we are then given a curve $Q$ and the task is to compute the curve $P \in S$ with smallest \frechet distance to $Q$. Since between two queries the curve $Q$ can change completely, one can show that a naive recomputation of all \frechet distances is (near-)optimal for exact algorithms. Research has thus focused on approximation algorithms~\cite{driemel2013jaywalking,DriemelPS19,Filtser18,FiltserF21,BringmannDNP22} and on restricted curves~\cite{de2017data,buchin2022efficient}.

\paragraph{Edit Distance}
The edit distance is a similarity measure on strings that counts the number of character insertions, deletions, and substitutions to transform one string into the other.
It can also be computed in time $O(n^2)$~\cite{NW70,Sel74,WF74}, but not in time $O(n^{2-\delta})$ for any $\delta > 0$ assuming the Strong Exponential Time Hypothesis~\cite{BI18}, even on binary strings~\cite{bringmann2015quadratic}. 
Edit distance on dynamically changing strings admits an exact algorithm with $\tilde{O}(n)$ update and query time~\cite{CKM20} (which is optimal assuming the Strong Exponential Time Hypothesis), and it admits an $n^{o(1)}$-approximation algorithm with $n^{o(1)}$ update and query time~\cite{KMS23}. 

Edit distance is usually studied with unit cost for insertions, deletions, and substitutions, but there is also a weighted variant in which the cost of any operation depends on the involved characters. The quadratic-time algorithm also works for this weighted edit distance. In the dynamic setting, weighted edit distance admits the same trade-off as we show for DTW in this paper, i.e., it has an exact algorithm with $\tilde{O}(n^{1+\beta})$ update time and $\tilde{O}(n^{2-\beta})$ query time for any $\beta \in [0, \frac{1}{2}]$~\cite[easily obtainable from Thm 11]{CKM20}, and there is a matching lower bound assuming the All-Pairs Shortest Paths hypothesis~\cite{CKW23}. 
Although these results are in the area of string algorithms, whereas here we study DTW, they are the most related work in the literature. In fact, there is a reduction from edit distance to DTW~\cite{kuszmaul2019dynamic}, so DTW can be viewed as the harder problem. Therefore, our dynamic DTW algorithm is a generalization of the dynamic weighted edit distance algorithm, following a similar approach. On the other hand, the lower bound for dynamic weighted edit distance in~\cite{CKW23} produces strings with $\Theta(n)$ different characters. Thus, when an update introduces a new character, then $\Theta(n)$ distances to the existing characters need to be specified. This setting cannot be modelled using the DTW distance of curves over $\Real^{d'}$ (for $d'=o(n)$), where each point can be described by $\Oh(d')$ coordinates, which is not enough to encode $\Theta(n)$ independent distances. We circumvent this issue by reducing from Negative-$k$-Clique instead of All-Pairs Shortest Paths. The resulting conditional lower bound follows the same high-level ideas as in~\cite{CKW23}, but needs to deviate considerably in the details, because we need gadgets that test for cliques instead of triangles.

\section{Preliminaries} \label{sec:preliminaries}
In this section we present the key concepts required for our results: rectangular graphs, distance measures between curves and conditional lower bounds. 
We denote by $\mathbb{N}$ the natural numbers, by $\mathbb{R}$ the reals, and for all $n \in \mathbb{N}$ by $[n]$ the integers $1$ up to (and including) $n$.
We use matrix indexing: a matrix with $n$ rows and $m$ columns is an $n \times m$ matrix. Row $i$ and column $j$ correspond to the point $(i, j)$. 
When drawing a grid graph in the plane, $(1, 1)$ indicates the top left vertex. 

\paragraph{Graphs and shortest paths.}
A graph $G = (V, E)$ is a set of vertices $V$ connected by edges $E$. 
The graph $G$ may be vertex-weighted (assigning every $v \in V$ some weight $\omega(v) \in \mathbb{R}$) or edge-weighted (assigning every  $e \in E$ some weight $\omega(e)\in \mathbb{R}$). 
A \emph{path} $\pi$ in $G$ is any sequence of unique vertices, such that consecutive vertices are connected by an edge in $E$.
When $G$ is vertex-weighted, the \emph{cost} of a path $\pi$ is $c(\pi) := \sum_{v \in \pi} \omega(v)$.
When $G$ is edge-weighted, \emph{cost} $c(\pi)$ is the sum over all edges $e$ between consecutive vertices in $\pi$ of $\omega(e)$. 
For any two vertices $u, v \in V$, we denote by $d_G(u, v)$ the minimum cost $c(\pi)$ over all paths $\pi$ that have $u$ and $v$ as its endpoints.

\paragraph{Rectangular graphs.}
A \emph{rectangular} graph $G$ has two integers $(A, B)$.
The vertex set $V$ contains $A \cdot B$ vertices: a vertex $(i, j)$ for all $(i, j) \in [A] \times [B]$. 
We can embed the graph on an integer grid by giving every vertex $(i, j) \in V$ coordinates $(i, j)$. 
The edge set $E$ is defined by the following edges:
\begin{itemize}[noitemsep]
    \item between $(i, j)$ and $(i + 1, j)$  \quad \quad\; (``vertical edges''),
\item between $(i, j)$ and $(i, j + 1)$  \quad \quad\;  (``horizontal edges''), and
\item between $(i, j)$ and $(i + 1, j + 1)$  \quad (``diagonal edges''). 
\end{itemize}
We intuitively refer to the columns and rows of $G$; $G$ is a plane embedded grid graph that contains $2(A + B) - 4$ vertices on the outer face. 
Rectangular (sub)graphs may be vertex-weighted or edge-weighted.
A path $\pi$ in a rectangular graph is $xy$-monotone, whenever in the sequence of vertices $\pi = \{ (i, j) \}$ are non-decreasing in both $i$ and $j$. 
For any two vertices $(i, j), (x, y) \in V$, we denote by $\overrightarrow{d}_G( (i, j), (x, y) )$ the cost of the cheapest $xy$-monotone path from $(i, j)$ to $(x, y)$.
Equivalently, the distance $\overrightarrow{d}_G( (i, j), (x, y) )$ is the cost of the cheapest path from $(i, j)$ to $(x, y)$ in the directed graph where each vertical edge points from $(i, j)$ to $(i+1, j)$ and so forth.  

A \emph{rectangular subgraph} $R$ of $G$ is defined by two intervals $[a, b] \subseteq [A]$, $[c, d] \subseteq [B]$. 
Its vertices are  $V_R := \{ (i, j) \in V \mid (i, j) \in [a, b] \times [c, d] \}$ and it contains an edge between two vertices in $V_R$ whenever they share an edge in $E$. 
Note that we do not reindex the vertices.
E.g., given a rectangular graph for $A = B = 10$, and the rectangular subgraph $R$ given by $[2, 4]$ and $[6, 9]$, the set $V_R$ contains the vertex $(4, 9)$, but not $(1, 1)$.

\paragraph{Discrete distance measures.}
Consider some metric space $X$. We denote for any pair $(p, q) \in X$ by $d(p, q)$ their distance. 
A \emph{curve} $P$ is any finite ordered sequence of points in $X$; we refer to points in $P$ as vertices of $P$. 
Any two curves $P  = (p_1, \ldots, p_n)$ and $Q = (p_1, \ldots, p_m)$ with $n$ and $m$ vertices respectively, induce a vertex-weighted rectangular graph $G$ with $A = n$ and $B = m$, where the weight of each vertex $(i, j)$ is $\omega( (i, j) ) = d(p_i, q_j)$. 

Given two curves $P$ and $Q$, we can define a distance measure to illustrate the similarity between $P$ and $Q$. 
There are two commonly used discrete similarity measures between curves $P$ and $Q$, each of which can be formalized using the vertex-weighted rectangular graph induced by $P$ and $Q$. 
The first measure is the discrete \frechet distance. 
Denote by $\Pi$ all $xy$-monotone paths with $(1, 1)$ and $(n, m)$ as their endpoints. 
For any $\pi \in \Pi$ the \emph{bottleneck cost} is $b(\pi) := \max_{(i, j) \in \pi} d(p_i, q_j)$. 
The discrete \frechet distance is subsequently defined as: 
\[
\FD (P, Q) := \min_{ \pi \in \Pi } b(\pi)  = \min_{ \pi \in \Pi } \max_{(i, j) \in \pi } d(p_i, q_j).
\]

\noindent
The second distance measure is the Dynamic Time Warping (DTW) distance: which is equal to the minimum cost of all $xy$-monotone paths from $(1, 1)$ to $(n, m)$ in $G$:
\[
\DTW (P, Q) := \min_{ \pi \in \Pi } c(\pi)   = \min_{ \pi \in \Pi } \sum_{(i, j) \in \pi} d(p_i, q_j).
\]

\paragraph{Upper bounds and (conditional) lower bounds.}
In the static problem variant we are given as input the metric space $X$ and two curves $P  = (p_1, \ldots, p_n)$ and $Q = (q_1, \ldots, q_m)$.
Given $P$ and $Q$ we may construct the corresponding rectangular graph $G$ using $O(nm)$ time and space. 
To compute our distances, we may perform depth-first search from the vertex $(1, 1)$ until we find the vertex $(n, m)$ (where the weight of an edge $( (i, j), (x, y))$ is the weight associated to $(x, y)$).
Both distances can be computed faster (in $O(nm)$ time) using dynamic programming. 

Whilst the above static algorithms are simple, there are conditional lower bounds showing that they are also optimal (up to subpolynomial factors).
For both the \frechet distance and DTW distance these conditional lower bounds can be based on SETH:

\begin{definition}
    The Strong Exponential Time Hypothesis (SETH) asserts that for any $\delta > 0$, there is an integer $k > 3$ such that $k$-SAT cannot be solved in $O(2^{ (1 - \delta) n})$ time.
\end{definition}

Consider any $\delta' > 0$ and a static algorithm that can compute the \frechet or the \DTW{} distance between curves $P$ and $Q$ in time $O( (nm)^{1 - \delta'})$. 
The conditional lower bounds by~\cite{abboud2015tight, bringmann2014walking, bringmann2015quadratic} show that such an algorithm (even for curves in $\mathbb{R}^1$ under any $L_p$ metric) would imply that SETH is wrong.

\paragraph{Dynamic DTW distance.}
In the dynamic problem variant, we have two curves $P  = (p_1, \ldots, p_n)$ and $Q = (q_1, \ldots, q_m)$ stored in some data structure subject to vertex insertions and deletions. 
We choose to formulate the problem using insertions, deletions and queries. 
A deletion may select either $P$ or $Q$, any vertex in the curve, and remove the vertex from the sequence. 
An insertion has as input a point $p$ and may select either $P$ or $Q$ and some vertex in the curve. It inserts $p$ into the curve after (or before) the selected vertex. 
We focus on the DTW distance, and require that a query reports the DTW distance between $P$ and $Q$. 

The conditional lower bounds for the static problem variant imply lower bounds for the dynamic version. Indeed consider any $\delta' > 0$ and a dynamic algorithm with $U(n, m)$ update time and $Q(n, m)$ query time such that $(n + m) \cdot U(n, m) + Q(n, m) \in O( (nm)^{1 - \delta'})$. This dynamic algorithm violates SETH, as we may answer the static problem variant using $(n + m)$ updates and a single query. 
In this paper, we provide a much stronger lower bound: not only presenting a higher lower bound but also restricting the query and update time individually.
We condition our lower bound on the Negative-$k$-Clique Hypothesis:

\begin{restatable}{definition}{kClique} \label{def:kClique}
In the Negative-$k$-Clique problem, the input is an undirected $k$-partite graph $G$ with $N$ nodes and integer edge weights, and the task is to decide whether there exists a $k$-clique in $G$ with a negative sum of all edge weights.
Let $W$ be the sum of the absolute values of all edge weights in $G$.
The Negative-$k$-Clique Hypothesis postulates that for all $\delta > 0$ the Negative-$k$-Clique problem for $W=N^{O(k)}$ cannot be solved in time $O( N^{k - \delta})$. 
\end{restatable}


\newcommand{\Rc}{\overline{\mathbb{R}}_{\ge 0}}

\section{Upper Bound} \label{sec:upper-bound}
In this section, we present a data structure to dynamically maintain the DTW distance between two curves $P$ and~$Q$.
Let $P$ be a curve with $n$ vertices, and let $Q$ be a curve with $m$ vertices, where $n \geq m$. Moreover, fix $\beta \in [0, \frac12]$. 
We prove Theorem~\ref{thm:upperboubnd} as we design a dynamic algorithm with $O(n m^{\beta} \log m)$ update time and $O(nm^{1 - \beta})$ query time.  
Before we state the algorithm and associated data structure, we state an auxiliary result that encapsulates Klein's multiple-source shortest path algorithm~\cite{Klein05} and 
the SMAWK algorithm~\cite{AggarwalKMSW87}.

Let us denote $\Rc=\mathbb{R}_{\ge 0}\cup\{\infty\}$ and recall that the min-plus product $x\star A$ of a row vector $x\in \Rc^X$ and a matrix $A\in \Rc^{X\times Y}$ is a row vector $y\in \Rc^Y$ such that $y[j] = \min_{i\in X}(x[i]+A[i,j])$ holds for all $j\in Y$.
\begin{lemma}\label{lem:kleinsmawk}
Let $G$ be a weighted plane digraph with $N$ vertices, and  let $S=\{s_1,\ldots,s_{|S|}\}$
and $T=\{t_1,\ldots,t_{|T|}\}$ be such that $(t_1=s_1,s_2,\ldots,s_{|S|}=t_{|T|},\ldots,t_2)$ is the cycle around the outer face.
Let $D\in \Rc^{S\times T}$ be a matrix such that, for every $s\in S$ and $t\in T$, the entry $D[s,t]$ encodes the shortest-path distance from $s$ to $t$.

There exists a data structure of size $\Oh(|S|\cdot |T|)$ that can be constructed in $\Oh((N+|S|\cdot |T|)\log N)$ time and, given a vector $x\in \Rc^S$, computes the min-plus product $x\star D\in \Rc^T$ in $\Oh(|S|+|T|)$ time.
\end{lemma}
\begin{proof}
Let $w(G)$ denote the total weight of all arcs in $G$.
For any real parameter $W>w(G)$, consider a plane digraph $G^{+W}$ obtained from $G$ by introducing a cost-$W$ backward arc $(v,u)$ along with every arc $(u,v)$ in $G$. 
Moreover, let $D^{+W}\in \Rc^{S\times T}$ be a matrix such that, for every $s\in S$ and $t\in T$, the entry $D^{+W}[s,t]$ encodes the shortest-path distance from $s$ to $t$ in $G^{+W}$.
Observe that all entries in $D^{+W}$ are finite because there are arcs in both directions between every two adjacent vertices around the outer face of $G$.
Moreover, when comparing the costs of paths in $G^{+W}$, the path containing fewer cost-$W$ edges is always cheaper. 
Consequently, 
\[D[s,t]=\begin{cases}
    D^{+W}[s,t] & \text{if }D^{+W}[s,t]<W,\\
    \infty & \text{otherwise.}
\end{cases}\]
and, for any other parameter $W'>w(G)$, we have
\[D^{+W'}[s,t] = D^{+W}[s,t]+(W'-W)\lfloor{D^{+W}[s,t]/W}\rfloor.\]

Our data structure consists of a value $W>w(G)$ (say, $w(G)+1$) and the matrix $D^{+W}$. 
It can be constructed in $\Oh((N+|S|\cdot |T|)\log N)$ time by a direct application of Klein's multiple-source shortest path algorithm~\cite{Klein05} on $G^{+W}$.
This algorithm, after $\Oh(N\log N)$-time preprocessing, allows $\Oh(\log N)$-time computation of any distance between a vertex on the outer face and any other vertex.

At query time, given $x\in \Rc^{S}$, we determine $W'=1+W+\max\{x[s] : x[s]\ne \infty\}$
and construct $x^{+W'}\in \Rc^{S}$ obtained from $x$ by replacing all infinite entries with $W'$.
Then, we compute the min-plus product $y^{+W'}:= x^{+W'}\star D^{+W'}$ and return a vector $y$ obtained from $y^{+W'}$ by replacing with $\infty$ all entries with values $W'$ or more.

As for correctness, observe that, for every $s\in S$ and $t\in T$, we have $x^{+W'}[s]+D^{+W'}[s,t]=x[s]+D[s,t]<W'$ if $x[s]\ne \infty$ and $D[s,t]\ne \infty$, and $x^{+W'}[s]+D^{+W'}[s,t]\ge W'$ if $x[s]=\infty$ or $D[s,t]=\infty$.
Consequently, $y^{+W'}[t]=y[t]$ if $y[t]\ne \infty$ and $y^{+W'}[t]\ge W'$ otherwise.

To design an efficient implementation, consider the vertices in $S\cup T$ and their cyclic order on the outer face of $G$. 
By definition of $S$ and $T$, the rows and columns of $D^{+W'}$ are ordered so that $s_1\prec \ldots \prec s_{|S|}$ and $t_1\prec \ldots t_{|T|}$, respectively, then $D^{+W'}$ satisfy the Monge property~\cite{Monge1781}; see~\cite[Section 2.3]{FR06}.
Thus, we can use the SMAWK algorithm~\cite{AggarwalKMSW87} to compute $x^{+W'}\star D^{+W'}$ in $\Oh(|S|+|T|)$ time given constant-time random access to $D^{+W'}$.
The relation between $D^{+W'}$ and $D^{+W}$ reduces constant-time random access to $D^{+W'}$ to constant-time random access to $D^{+W}$ (which is a part of our data structure).
\end{proof}

\paragraph*{Data structure and query algorithm.}
Having established our prerequisites, we present an overview of our data structure. 
Recall that the two curves $P$ and $Q$ induce a vertex-weighted $n \times m$ rectangular graph, where the weight of the vertex $(i, j)$ is the distance $d(p_i, q_j)$. 
We denote this graph by $G$. 
The DTW between $P$ and $Q$ is the minimal cost $xy$-monotone path from $(1, 1)$ to $(n, m)$ in $G$.
We fix a parameter $\beta \in [0, \frac{1}{2}]$ and show that we can efficiently compute this minimal cost path through three steps.
Our data structure is illustrated by Figure~\ref{fig:datastructure}.

Our first step (Lemma~\ref{lem:partition}) is that, if $n \geq m$, we dynamically maintain a partition $\mathbf{P}$ of $P$ into $O(n m^{-\beta})$ subcurves, where each subcurve has at least $\frac{1}{4} m^\beta$ and at most $4 m^\beta$ vertices. 
Similarly, we maintain a partition $\mathbf{Q}$ of $Q$ into $O(m^{1-\beta})$ subcurves with a size in $[\frac{1}{4} m^\beta, 4 m^\beta]$.

\begin{figure}[t]
  \centering
  \includegraphics[scale = 1.2]{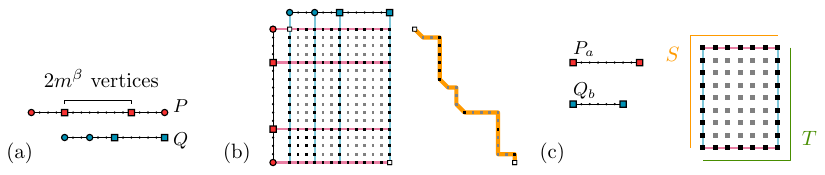}
  \caption{
(a) We partition $P$ and $Q$ into subcurves of $\Theta(m^\beta)$ vertices.
(b) This induces a partition of the grid graph into $O(nm^{1 - 2\beta})$ rectangles. 
(c) For each pair of subcurves $(P_a, Q_b)$, we consider the vertices $S$ on the top and left facet and the vertices $T$ on the bottom and right facet.
    }
 \label{fig:datastructure}
\end{figure}

Secondly, we note that each of the $O(nm^{1 - 2\beta})$ pairs of subcurves $(P_a, Q_b) \in \mathbf{P} \times \mathbf{Q}$ correspond to a rectangular subgraph $R^{ab}$ in $G$ with $O(m^\beta)$ points on its boundary and $O(m^{2\beta} )$ points in its interior. 
Let $S$ be the vertices on the left and top boundary of $R^{ab}$ and $T$ be the vertices on the right and bottom boundary of $R^{ab}$.
Intuitively, we would like to build a matrix such that, for $s \in S$ and $t \in T$,
the entry $A^{ab}[s, t]$ stores the distance~$\overrightarrow{d}_{R^{ab}}(s, t) = \overrightarrow{d}_{G}(s, t)$. 
For technical reasons, we transform the vertex-weighted rectangular graph $R^{ab}$ with an edge-weighed \emph{alignment graph} $D^{ab}$ such that, for each pair of vertices $(s, t)$  in $R^{ab}$, the length of the shortest path from $s$ to $t$ in $D^{ab}$ uniquely corresponds to the length of the shortest $xy$-monotone path from $s$ to $t$ in $R^{ab}$.
The distances from $S$ to $T$ in $D^{ab}$ are stored using the data structure of Lemma~\ref{lem:kleinsmawk}.
There are $O(nm^{1-2\beta})$ rectangular subgraphs $R^{ab}$, each of which have $O(m^{2\beta})$ edges and vertices. 
For each $R^{ab}$, we store the alignment graph $D^{ab}$ in a data structure of size $O( m^{2\beta} \log m)$; thus, our data structure takes $O(nm \log m)$ space. 

Finally, we show that this data structure allows us to compute the DTW distance between~$P$ and~$Q$ in $O(n m^{-\beta})$ time through the following ``wavefront'' algorithm (Figure~\ref{fig:wavefront}):
Consider the vertex-weighted grid graph $G$ between $P$ and $Q$. 
We compute for all $x \in [n]$ and all $y \in [m]$ the distances $\overrightarrow{d}_G ( (1, 1), (x, 1)) $ and $\overrightarrow{d}_G ( (1, 1), (1, y)) $.
We call this set of $O(n)$ values the wavefront~$\mathcal{W}$. 
Throughout the algorithm, we maintain a wavefront~$\mathcal{W}$ of size $O(n)$ and store for each point $(i, j) \in \mathcal{W}$ the value $\overrightarrow{d}_G ( (1, 1), (i, j))$. 
We iteratively update the wavefront as follows: 
as long as the point $(n, m)$ is not in the wavefront, there always exists at least one rectangle $R^{ab}$ whose left and top facet coincide with the wavefront. 
We select one such rectangle $R^{ab}$, remove its left and top facet $S$ from the wavefront and replace them with the bottom and right facet $T$. 
We show that we can perform this operation in $O(m^\beta)$ time by querying the data structure of \cref{lem:kleinsmawk}.
After $O(nm^{1 - 2\beta})$ iterations, we add the point $(n, m)$ to the wavefront and we know the length of the shortest $xy$-monotone path from $(1, 1)$ to $(n, m)$. Thus, given our data structure, we compute the DTW distance between $P$ and $Q$ in $O(n m^{- \beta} \log m)$ time. 
In the remainder of this section we formalise each data structure component.

\paragraph*{Dynamic partitions.} We dynamically maintain a partition of $P$ and a partition of $Q$ into subcurves of size $O(m^\beta)$ under very specific conditions: 

\begin{lemma}
\label{lem:partition}
    Let $P$ and $Q$ be curves where their lengths are $n_0$ and $m_0$ before receiving any updates. 
    Let $\mathbf{P}$ and $\mathbf{Q}$ be partitions of $P$ and $Q$ respectively where each subcurve has at least $m_0^\beta$ and at most $2 m_0^\beta$ vertices. 
    During a sequence of $\frac{m_0}{2}$ updates to $P$ or $Q$, after which $Q$ has $m$ vertices,  we can dynamically maintain $\mathbf{P}$ and $\mathbf{Q}$ such that each subcurve has a size in $[\frac{1}{4} m^\beta, 4 m^\beta]$, using $O(n + m)$ space and $O(m_0^\beta)$ time per operation. Moreover, our updates change at most $O(1)$ subcurves of $P$ and $Q$. 
\end{lemma}

\begin{proof}
  We show how to update $P$ after inserting/deleting a vertex in $p$. 
   Updates in $Q$ are handled analogously. 
   Denote by $m$ the size of $Q$ during updates.
    Then at all times, we have that $m \in [\frac{1}{2} m_0, \frac{3}{2} m_0]$. 
    We store for each subcurve its two boundary vertices and its size, and we store all subcurves in a balanced binary tree sorted by size. 
    Each vertex stores a pointer to the subcurve that contains it. 
    Finally, we store the numbers $n$ and $ m^\beta_0$. 
 Suppose that we insert a vertex $p$, preceding a vertex $p_i$ that  lies in the subcurve $P_a$, or we delete a vertex $p$ lying in a subcurve $P_a$. 
    We add/remove $p$ to $P_a$, incrementing/decrementing the size of the subcurve. 
    If $p_i$ was the left boundary vertex of $P_a$, we make $p$ the left boundary vertex. 
    We update our balanced binary tree in $O(\log n)$ time. 

    If we add a vertex $p$ to $P$, we add it to the subcurve $P_1 \in \mathbf{P}$ that contains its successor on $P$. 
    If the size of $P_1$ is larger than $2  m_0^\beta $, we spend $O(m_0^\beta)$ time to split $P_1$ into two subcurves of roughly equal size (by iterating over all vertices in $P_1$ and selecting the median). 
    If, after deleting a vertex $p$ from $P$ (and thereby from $P_1 \in \mathbf{P}$) the size of $P_1$ is smaller than $ \frac{1}{2}  m^\beta_0 $, we consider an arbitrary subcurve $P_2$ incident to $P_1$ and join the two subcurves. The resulting curve $P'$ must have length at most $ 2.5  m_0^\beta $. 
 If $P'$ is longer than $2  m^\beta_0 $ we split $P'$ along its median: creating two subcurves whose length lie in  $[\frac{1}{2} m_0^\beta, 2 m_0^\beta]$.
    
    If we update $Q$ we may change $m$. 
    For all curves $P_a \in  \mathbf{P}$ we showed that their size remains in $[\frac 12 m_0^\beta, 2m_0^\beta] \subseteq [\frac{1}{4}m^\beta, 4m^\beta]$.
\end{proof}

\paragraph*{Dynamically storing distance matrices.}

The partitions $\mathbf{P}$ and $\mathbf{Q}$ partition our vertex-weighted grid graph $G$ into rectangles. 
Each pair of subcurves $(P_a, Q_b)$ induces a rectangular grid graph $R^{ab}$ where we want to store the $xy$-monotone distance matrix $A^{ab}$ between all boundary vertices of $R^{ab}$. 
We denote by $\overrightarrow{d}_{G}( (i, j), (x, y) )$ the cost of the cheapest $xy$-monotone path from $(i, j)$ to $(x, y)$ in $G$ ($\overrightarrow{d}_{G}( (i, j), (x, y) ) = \infty$ if no such path exists).

\begin{definition}[Alignment Graph]
Let $P_a = (p_\alpha, \ldots, p_\delta)$ and $Q_b = (q_\gamma, \ldots, q_\nu)$ be two curves with $N$ and $M$ vertices,
respectively.
We define the alignment graph $D^{ab}$ as a rectangular graph with $N \times M$ vertices and the following edges:
\begin{itemize}[noitemsep]
\item $(i, j)\to (i + 1, j)$ of weight $d(p_{i + 1}, q_j)$ (``vertical edges''),
\item $(i, j)\to (i, j + 1)$ of weight $d(p_i, q_{j + 1})$ (``horizontal edges''),
\item $(i, j)\to (i + 1, j + 1)$ of weight $d(p_{i + 1}, q_{j + 1})$ (``diagonal edges''). 
\end{itemize}
\end{definition}

\begin{lemma}
\label{lem:equivalentpaths}
Let $P_a = (p_\alpha, \ldots, p_\delta)$ and $Q_b = (q_\gamma, \ldots, q_\nu)$ be two curves with $N$ and $M$ vertices.
Denote by $R^{ab}$ and $D^{ab}$ their rectangular and alignment graph, respectively. 
For $p_i, p_x \in P_a$ and $q_j, q_y \in Q_b$, the cost $\overrightarrow{d}_{G}( (i, j), (x, y) )$ is equal to:
\[
\overrightarrow{d}_{R^{ab}} ( (i, j), (x, y))  =  d_{D^{ab} }( (i, j), (x, y))  + d(p_i, q_j).
\] 
\end{lemma}
\begin{proof}
Any $xy$-monotone path from $(i, j)$ to $(x, y)$ in $G$ must be contained in the rectangular graph~$R^{ab}$.
Thus, there exists a bijection between $xy$-monotone paths from $(i, j)$ to $(x, y)$ in $G$ and in~$D^{ab}$.
The cost of any $xy$-monotone path from $(i, j)$ to $(x, y)$ in $G$, equals the cost of the uniquely corresponding $xy$-monotone path from $(i, j)$ to $(x, y)$ in $D^{ab}$ (plus $d(p_i, q_j)$).
Moreover, the orientation of edges guarantees that all paths in $D^{ab}$ are monotone.
\end{proof}
\noindent

 \noindent
 Having established our alignment graph $D^{ab}$, we are ready to define our update procedure.

\begin{lemma}
\label{lem:datastructure}
    Let $P$ and $Q$ be two dynamic curves and assume that $|P| = n \geq |Q| = m$. 
    We can maintain a partition of $P$ and $Q$ into subcurves with $\Theta(m^\beta)$ vertices each where, for all $(P_a,Q_b)$, we store the graph $D^{ab}$, with sources $S$ on the left and top boundary and targets $T$ on the right and bottom boundary, using the data structure of \cref{lem:kleinsmawk}.
    Our data structure requires $O(n m)$ space and has $O(n m^{\beta} \log m)$ update time. 
\end{lemma}
\begin{proof}
    First, we describe our data structure. 
    We want to, at all times, maintain a pointer to the following data structure that stores partitions $\mathbf{P}$ and $\mathbf{Q}$ of $P$ and $Q$, respectively, into subcurves that have a size in $[\frac{1}{4}m^\beta, 4m^\beta]$. 
    For all $O(n m^{1 - 2\beta})$ pairs of subcurves $P_a$ and $Q_b$, the rectangular subgraph $R^{ab}$ of $G$ has size $O( m^{2\beta})$.
    Our data structure stores $D^{ab}$, with sources $S$ on the left and top boundary and targets $T$ on the right and bottom boundary, using the data structure of Lemma~\ref{lem:kleinsmawk}.
    This requires $O(m^{2\beta})$ space and $O(m^{2\beta}\log (m^\beta))$ time to construct per subgraph $R^{ab}$.
    Thus, the total space used is $O(nm)$ and we may construct this data structure in $O(nm \log m)$ total time.
    
    We describe our update strategy.
    For each update we increment a counter $c$ by $1$. 
    Whilst $c \leq m_0/2$, we may dynamically maintain $\textbf{P}$ and $\textbf{Q}$ by performing updates such that each subcurve $P_a$ and $Q_b$ has $O(m^\beta)$ vertices (Lemma~\ref{lem:partition}). 
    During every such update, by Lemma~\ref{lem:partition}, at most $O(1)$ subcurves $P_a \in \textbf{P}$ and $Q_b \in \textbf{Q}$ change. 
    Whenever we change a subcurve $Q_b$  (e.g., the subcurve $Q_b$ lost a vertex, or is obtained by splitting a previous subcurve along its median) we do the following:  
    for all $O(n m^{-\beta})$ subcurves $P_a$, we consider the rectangular subgraph $R^{ab}$ of $G$.
    This graph has $O(m^{2\beta})$ weighted vertices. 
    We construct the corresponding alignment graph in $O( m^{2\beta})$ time and apply the construction algorithm of Lemma~\ref{lem:kleinsmawk} in $O(m^{2\beta} \log m)$ time.
    Since at most $O(1)$ subcurves $P_a$ and $Q_b$ change, each update takes $O( n m^{-\beta} m^{2\beta} \log m) = O(n m^\beta \log m)$ total time. 

    Given any $(P, Q)$, $n_0$ and $m_0$, by our above reasoning, we may statically construct partitions $\mathbf{P}$ and $\mathbf{Q}$ where subcurves have a size in $[m_0^\beta, 2 m_0^\beta]$ (and the associated data structure) in $O( n_0 m_0 \log m_0)$ time. Hence, when the counter reaches $c = m_0/2$, we can rebuild the data structure in $O(n_0 m_0 \log m_0)$ time. This yields \emph{amortized} update time $O(n m^\beta \log m)$. 

    In what follows we apply a classic deamortization scheme, to prove the lemma.
    We maintain at all times the above data structure twice, referring to them as the first and second copy. 
    Each copy stores a counter, $c_0$ and $c_1$ respectively, that counts the number of updates processed by each copy. 
    At all times, we maintain a pointer to one of the two copies, indicating the current `up to date' data structure. 
    We denote by $n_0$ and $m_0$ the initial size of $P$ and $Q$ respectively (before any updates) and by $n$ and $m$ the current size of $Q$. 
    We assume that the first copy has, before receiving any updates, $P$ and $Q$ partitioned into subcurves of size $[m_0^\beta, 2 m_0^\beta]$ and that is has recorded the value $m_0$. 
       
    Our deamortization scheme ensures that we always perform fewer than $m_0/2$ updates to the first copy. 
    Our counter $c_0$ starts at $\frac{8 m_0}{32}$. 
    We note that for readability, we over-estimate our constants to be able to write them as multiples of two. 
    When $c_0 = \frac{9 m_0}{32}$, we record the value $m_1 = m$ and store it in the second copy.
    Note that $m_1 \in [\frac{1}{2}m_0, \frac{6}{4} m_0]$.     
    In addition, we record the curves $P^1 = P$, $Q^1 = Q$, $n_1 = n$, and $c_1 = 0$.
    From this point onwards, we start recording updates to the first copy in a queue. 
    
    Whilst $c_0 \in [\frac{9 m_0}{32}, \frac{10 m_0}{32}]$, we construct as our second copy our data structure on $(P^1, Q^1)$ in $O(n_1 m_1 \log m_1)$ total time, doing $\Theta(n_1 \log m_1) = \Theta(n \log m)$ work per update.
    When $c_0 = \frac{10 m_0}{32}$, the queue of the first copy contains at most $\frac{m_0}{32} \leq \frac{2 m_1}{32}$ elements. 
    From hereon, each time time $c_0$ is incremented, we perform an update in the first copy, add it to the queue, dequeue up to four updates from the queue and apply them to the second copy (incrementing $c_1$ by four). 
    When $c_0 = \frac{11 m_0}{32}$, both the first and second copy store the same data structure. Moreover, $c_1 \leq 4 \cdot \frac{m_0}{32} \leq \frac{8 m_1}{32}$.
    We continue applying all updates to both data structures (incrementing $c_1$ and $c_0$ by $1$) until $c_1 = \frac{8 m_1}{32}$. 
    
    At this point, we record $P^0 = P, Q^0 = Q, m_0 = m, n_0 = n$ and set $c_0 \gets 0$. 
    We note that $m_0 \in  [\frac{1}{2}m_1, \frac{6}{4} m_1]$.
    From hereon, we perform the process with the two copies exchanged.  Since at all times, $c_0 \leq \frac{m_0}{2}$ and $c_1 \leq \frac{m_1}{2}$, we may always apply Lemma~\ref{lem:partition} to perform our $O(1)$ updates in $O(n m^\beta \log m)$ time. 
\end{proof}

\paragraph{Computing the DTW distance.}

Finally, we are ready to show our main theorem:

\tradeoff*

\begin{proof}
We store $P$ and $Q$ in the data structure of Lemma~\ref{lem:datastructure} which has the desired space usage and update time. 
What remains is to show that we can compute the DTW distance between $P$ and $Q$.
Consider the rectangular graph $G$ and the partition $\mathbf{P} = (P_1, \ldots P_N)$ and $\mathbf{Q} = (Q_1, \ldots, Q_M)$. 
Consider the set of all rectangular subgraphs $R^{ab}$ for $P_a$ and $Q_b$ with $a  \in [N]$ and $b \in [M]$.
In $O(n)$ time, we compute for every integer $i \in [n]$ the cost of the vertical path from $(1, 1)$ to $(i, 1)$. 
Similarly, for each $j \in [m]$ we compute the cost of the horizontal path from $(1, j)$. 
We denote these vertices of $G$ as the ``wavefront'' $\mathcal{W}$. 
(Note that, whilst our paths are $xy$-monotone curves that are increasing, the wavefront is a decreasing $xy$-monotone curve.)
Throughout our algorithm, we maintain the invariant that for each vertex $(i, j) \in \mathcal{W}$, we store the value $\overrightarrow{d}_G( (1, 1), (i, j))$.

We iteratively expand $\mathcal{W}$ as follows. 
At each iteration, there exists at least one rectangular graph $R^{ab}$ whose left and top facets coincide with $\mathcal{W}$. 
Denote by $S$ all vertices on the left and top facet of $R^{ab}$ and by $T$ all vertices on the right and bottom facet.
We remove all $(a, b) \in S$ from $\mathcal{W}$, and add all $(x, y) \in T$ to $\mathcal{W}$. 
This ensures that $\mathcal{W}$ remains an $xy$-monotone curve.
To satisfy our invariant, we need to compute a vector where each coordinate corresponds to a point $(x, y) \in T$ and where the value at that coordinate records $\overrightarrow{d}_G( (1, 1), (x, y))$.

Observe that any $xy$-monotone path in $G$ from $(1, 1)$ to $(x, y) \in T$ must go through a vertex $(i, j) \in S$.
Thus, the length of the  shortest $xy$-monotone path in $G$ from $(1, 1)$ to $(x, y) \in T$ is equal to:

\begin{flalign*}
    &\overrightarrow{d}_G( (1, 1), (x, y) ) \\
    &\qquad= \min_{ (i, j) \in S}  \left(\overrightarrow{d}_G( (1, 1), (i, j)) -d(p_i,q_j) +\overrightarrow{d}_{R^{ab}}( (i, j), (x, y) ) \right) \\ 
    &\qquad=  \min_{ (i, j) \in S } \left( \overrightarrow{d}_G( (1, 1), (i, j))+ d_{D^{ab}} ( (i, j),(x, y)) \right).
\end{flalign*}
Given this relation between distances in $G$ and distances in our alignment graphs $D^{ab}$, we can compute our desired output by applying Lemma~\ref{lem:kleinsmawk} for the vector assigning $ \overrightarrow{d}_G( (1, 1), (i, j))$ to each $(i,j)\in S$.

In $O(m^\beta)$ time, we iterate over each $(x, y) \in T$, for which there exists a unique entry in the output vector that stores the value $\overrightarrow{d}_G( (1, 1), (x, y) ) =  \min_{ (i, j) \in S } \left( \overrightarrow{d}_G( (1, 1), (i, j)) + d_{D^{ab}} ( (i, j), (x, y)) \right)$, and add $(x, y)$ to our wavefront.
It follows that, in $O(m^\beta)$ time, we processed $R^{ab}$, removing $S$ from the wavefront, adding $T$ and maintaining our invariant. 

After $O( nm^{1-2\beta} )$ iterations (taking $O( nm^{1-\beta})$ total time), we process the last rectangle $R^{NM}$ and thus add the point $(n, m)$ to our wavefront. 
Via our invariant, we have computed the shortest $xy$-monotone path in $G$ from $(1, 1)$ to $(n, m)$ and therefore the DTW distance between $P$ and $Q$.
\end{proof}

\begin{figure}[t]
  \centering
  \includegraphics[width = \linewidth]{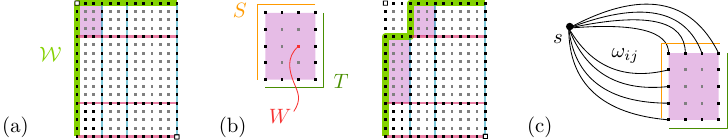}
  \caption{
(a) Consider our rectangular graph $G$, partitioned into rectangular subgraphs. Our wavefront $\mathcal{W}$ starts as a $\Gamma$-shape that includes $(1, 1)$. 
Each iteration, there exists at least one rectangular subgraph $R^{ab}$ (purple) whose left and top facet coincide with $W$.
(b) Given $R^{ab}$, we remove $S$ from the wavefront and add $T$.
(c) We want to compute a vector that records for all $(x, y) \in T$, the distance $ \overrightarrow{d}_G( (1, 1), (x, y) )$. 
To this end, we add a source $s$ that we connect to all $(i, j) \in S$ with an edge with weight $\omega_{ij}$ and apply Lemma~\ref{lem:kleinsmawk}.
    }
 \label{fig:wavefront}
\end{figure}

\section{Reducing from Intermediary to Dynamic DTW}
\label{sec:lower-bound-curve}
In this section, we study the \textsc{Intermediary} problem. 
We note that to better match previous results, we index from $0$ to $(n-1)$.

\intermediary*

For any instance of \textsc{Intermediary}, we show that one may maintain two curves $P$ and $Q$, where $P$ has $n \in O(n_r)$ vertices and $Q$ has $m \in O(n_c)$ vertices, so that every update in \textsc{Intermediary} corresponds to changing the position of four vertices in $Q$. 
Our curves are created in such a way that we may compute from $\DTW(P, Q)$ the output of \textsc{Intermediary} in $O(1)$ time.

\paragraph{The reduction}
For a fixed instance of \textsc{Intermediary}, our construction (Figure~\ref{fig:karlconstruction}) takes place on the real line and maps every row $i$ to a curve $\alpha_i$ and every column $j$ to a curve $\beta_j$. 
The curve $P$ is simply the concatenation over rows $i = 0$ to $(n_r - 1)$ of $\alpha_i$.
The curve $Q$ is the concatenation over columns $j = 0$ to $(n_c - 1)$ of $\beta_j$.
Note that an update $\textsc{Update}(j, x)$ in \textsc{Intermediary} then corresponds to translating all vertices in $\beta_j$ to the vertices of the new curve $\beta_j'$.
Hence, any update in \textsc{Intermediary} is realized by $O(1)$ translations in $Q$.

\begin{definition}[see \cref{fig:karlconstruction} -- left]
 Denote by $\star$ the point $-U^5\in \Real$.
 Denote by $\oEight{}$ a curve that visits the point $\star$ eight times consecutively. 
    Every row $i$ in \textsc{Intermediary} defines the curve $\alpha_i$:
    \begin{align*}
    \oEight{} \rightarrow & \quad \quad \quad \alpha_i^1 \quad  &\rightarrow   & \quad \quad \quad \alpha_i^2 
  &\rightarrow   & \quad \quad \quad \alpha_i^3 \quad &\rightarrow  & \quad \quad \alpha_i^4 &\rightarrow \oEight{} =  \\
   \oEight{} \rightarrow & \quad U^4 + 2r_i U^3  + \tfrac{d_i}{4} &\rightarrow  & \quad 2U^4 + 2 r_i U^3 - \tfrac{d_i}{4}  &\rightarrow  &\quad 3U^4 - 2 r_i U^3 + \tfrac{d_i}{4}  &\rightarrow  & \quad 4U^4 - 2 r_i U^3 - \tfrac{d_i}{4}  &\rightarrow \oEight{}
    \end{align*}
    Every column $j$ in \textsc{Intermediary} defines a curve $\beta_j$:
        \begin{align*}
    \oEight{} \rightarrow & \quad \quad \quad \beta^1_j \quad  &\rightarrow   & \quad \quad \quad \beta^2_j 
  &\rightarrow   & \quad \quad \quad \beta^3_j \quad &\rightarrow  & \quad \quad \beta^4_j &\rightarrow \oEight{} =  \\
    \oEight{} \rightarrow &   U^4 + 2 c_j U^3 - U &\rightarrow  &  2U^4 + 2 c_j U^3 + U 
 &\rightarrow &  3U^4 - 2 c_j U^3 + (-1)^{b_j} U &\rightarrow &  4U^4 - 2 c_j U^3 - (-1)^{b_j} U  &\rightarrow \oEight{}
    \end{align*}
\end{definition}

\begin{definition}
    We denote by $P$ the curve obtained by concatenating, over all rows $i\in [n_r]$, the curves $\alpha_i$. 
    We denote by $Q$ the curve obtained by concatenating, over all columns $j\in [n_c]$, the curves $\beta_j$.
    We denote by $R$ the $O(n_r) \times O(n_c)$ vertex-weighted rectangular graph induced by $(P, Q)$, as defined in \cref{sec:preliminaries}.
\end{definition}

\begin{definition}[see \cref{fig:karlconstruction} -- middle]
\label{def:subgadget}
For any $i, j$, the curves $\alpha_i$ and $\beta_j$ induce a $20 \times 20$ vertex-weighted rectangular graph which we call the \emph{gadget} $R_{ij}$. 
Each gadget $R_{ij}$ is a subgraph of $R$.
We call vertices incident to the boundary facets of $R_{ij}$ the \emph{boundary} vertices.
Each pair $(x, y)$ of vertices in $\alpha_i \times \beta_j$ corresponds to a vertex in $R_{ij}$.
We assign these vertices a color as follows:
\begin{itemize}[noitemsep]
    \item If $x = y = \star$ then the vertex is orange. 
    \item If either $x$ or $y$ equals $\star$ (but not both) then the vertex is white. 
        \item If $x = \alpha_i^{k}$ and $y = \beta_j^{k}$ for some $k \in [4]$, then the vertex is grey.
    \item Otherwise, the vertex is yellow. 
\end{itemize}

\begin{observation}
    \label{obs:minimise}
    Any $xy$-monotone path that realises $\DTW(P, Q)$ uses as few white vertices as possible, then as few yellow vertices as possible and finally as few grey vertices as possible. 
\end{observation}

\end{definition}

\begin{figure}[t]
  \centering
  \includegraphics[]{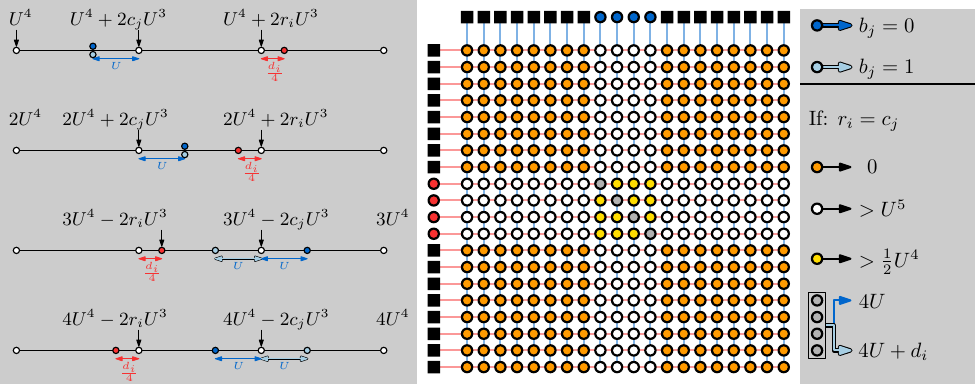}
  \caption{
For every row $i$, we place four red points on the real line $\Real$ (depending on $(r_i, d_i)$).
For every column $j$, we place four blue points (depending on $(c_j, b_j)$). 
Any pair of subcurves $(\alpha_i, \beta_j)$ induces a rectangular gadget $R_{ij}$. 
We show the vertex weights for when $r_i = c_j$. 
Whenever $r_i \neq c_j$, the only change is that grey vertices have weights in $[U^3, U^4]$. 
    }
 \label{fig:karlconstruction}
\end{figure}

\paragraph{Reducing from \textsc{Intermediary}.}
We show the following desirable property of our curves $P$ and $Q$:

\begin{lemma}
\label{lem:nowhites}
    For our curves $P$ and $Q$, there exists an $xy$-monotone path $\pi^*$ realizing $\DTW(P, Q)$ that contains no white boundary vertices. 
\end{lemma}

\begin{proof}
    For a proof by contradiction, suppose that every path realizing $\DTW(P,Q)$ contains a white boundary vertex.
    Let us fix a path $\pi$ that visits the fewest such vertices.
    First, suppose that $\pi$ visits a white boundary vertex $u$ located on the boundary of the entire graph $R$.
    By symmetry, we may assume that $u$ lies in the first row of $R$. Let $(2,y)$ be the first vertex on $\pi$ that lies in the second row of $R$, and let $\pi_1$ be the prefix of $\pi$ from the origin $(1,1)$ to $(2,y)$.
    Consider the following alternative path 
    \[\pi'_1 : (1,1) \to (2,2) \to (2,3) \to \cdots \to (2,y).\]
    Observe that, for each column $y'\in [y]$, the cost of $(1,y')$ is the same as the cost of $(2,y')$, and $\pi_1$ must visit at least one of these two vertices. 
    Consequently, $\pi_1$ is at least as expensive as $\pi'_1$.
    At the same time, $\pi'_1$ avoids white boundary vertices, whereas $\pi_1$ contains at least one such vertex ($u$).
    Thus, by replacing $\pi_1$ by $\pi'_1$, we transform $\pi$ into a path $\pi'$ 
that contains fewer white boundary vertices, contradicting the choice of $\pi$.

    Henceforth, we may assume that $\pi$ contains a white boundary vertex that is not located on the boundary of the entire graph $R$.
    Let us take the first such vertex $u$ (along $\pi$).
    Let $B$ be the connected component of $u$ in the subgraph of $R$ spanned by white boundary vertices;
    note that $B$ is a box spanning two rows and four columns (or, symmetrically, spanning four rows and two columns).
    Moreover, let $W$ be the maximum box of white vertices containing $B$ (it spans 16 rows and 4 columns, or 4 rows and 16 columns).
    
    Let $t$ be the last vertex of $\pi$ that lies above or to the left of $W$,
    and let $v$ be the first vertex of $\pi$ that lies below or to the right of $W$.

    If $t$ is above $W$ and $v$ is to the right of $W$, then the $t\leadsto v$ subpath of $\pi$ can be rerouted along the row just above $W$ and the column just to the right of $W$ (we use a diagonal edge whenever we switch from a row to a column or vice versa).
    Such a detour does not contain any white vertices, so it is cheaper than the original path, contradicting the minimality of $\pi$ (Observation~\ref{obs:minimise})

    If $t$ is to the left of $W$ and $v$ is below $W$, then the $t\leadsto v$ subpath of $\pi$ can be rerouted along the column just to the left of $W$ and row just below $W$.
    Again, such a detour does not contain any white vertices, so it is cheaper than the original path, contradicting the minimality of $\pi$.

    If $t$ is above $W$ and $v$ is below $W$, then the $t\leadsto v$ path contains at least $16$ white vertices (one per row of $W$).
    In this case, let $s$ be the last vertex of $\pi$ that lies to the left of $W$.
    Since $u$ was the first white boundary vertex on $\pi$, then $s$ must be located within the same gadget $R_{ij}$ as the upper half of $W$.
    Consequently, the $s\leadsto v$ path that goes along the column just to the left of $W$ and then along the row just below $W$ contains at most $4$ internal white vertices (one for each of the middle four rows of $R_{ij}$). Such a detour is thus cheaper than the original path,  contradicting the minimality of $\pi$.

    Finally, suppose that $t$ is to the left of $W$ whereas $v$ is to the right of $W$.
    In this case, let $s$ be the last vertex of $\pi$ that lies above $B$ or to the left of $W$.
    Since $u$ was the first white boundary vertex on $\pi$, then $s$ must be located with the same gadget $R_{ij}$ as the upper half of $B$, or within the right half of the adjacent gadget $R_{i(j-1)}$.
    Consequently, the $s\leadsto v$ path that goes along the row of $s$
    and along the column just to the right of $W$ contains exactly 4 interval vertices of positive cost: one white vertex per column of $W$.
    However, the original $s\leadsto v$ subpath of $\pi$ must have also contained such 4 white vertices.
    The costs of white vertices within $W$ are uniform along columns, so the detour is not more expensive.
    At the same time, the detour avoids $B$ (and thus any white boundary vertices) whereas the original path contained $u$. This contradicts the definition of $\pi$.
\end{proof}

\begin{definition}
\label{def:blocks}
Consider our curves $P$ and $Q$ and their induced rectangular graph. 
We define the \emph{blocks} (denoted by $\mathbf{B}$) of this graph as all maximal connected components of orange vertices. 
Two blocks $B_1, B_2 \in \mathbf{B}$ are:
\begin{itemize}[noitemsep]
    \item Horizontally adjacent if there exists a horizontal line that intersects $B_1$ and $B_2$ consecutively.
    \item Vertically adjacent if there exists a vertical line that intersects $B_1$ and $B_2$ consecutively.
    \item Diagonally adjacent if they are not horizontally/vertically adjacent and there exists a line with slope $-1$ that intersects $B_1$ and $B_2$ consecutively.
\end{itemize}
\end{definition}

\begin{lemma}
    \label{lem:horizontal}
    Let $B_1, B_2 \in \mathbf{B}$ be two blocks that are horizontally (or vertically) adjacent. 
    Then for any $u \in B_1$ and $v \in B_2$ the shortest $xy$-monotone path from $u$ to $v$ has weight $4U^5 + 10 U^4$. 
\end{lemma}

\begin{proof}
    If $(B_1, B_2)$ are horizontally adjacent then any shortest $xy$-monotone path from $u$ to $v$ consists of orange vertices plus exactly four white vertices corresponding to pairs:
    $(\star, \beta_j^1)$, $(\star, \beta_j^2)$, $(\star, \beta_j^3)$, $(\star, \beta_j^4)$ for some integer $j$.
    
    Since orange vertices have weight zero, it follows that the weight of this path is:
    \begin{align*}
         d(\star, \beta_j^1) + d (\star, \beta_j^2) + d(\star, \beta_j^3) + d(\star, \beta_j^4) = \\
         \left(U^5 + U^4 + 2 c_j U^3 - U \right) +  \left(U^5 + 2U^4 + 2 c_j U^3 + U \right) + \\
         \left(U^5 + 3U^4 - 2 c_j U^3 + (-1)^{b_j} U \right) + \left(U^5 + 4U^4 - 2 c_j U^3 + (-1)^{b_j} U \right) = 4 U^5 + 10 U^4.
    \end{align*}

If $(B_1, B_2)$ are vertically adjacent then any shortest $xy$-monotone path from $u$ to $v$ consists of orange vertices plus exactly four white vertices corresponding to pairs:
    $(\alpha_i^1, \star)$, $(\alpha_i^2, \star)$, $(\alpha_i^3, \star)$, $(\alpha_i^4, \star)$ for some integer $i$.
      Since orange vertices have weight zero, it follows that the weight of this path is:
    \begin{align*}
         d(\alpha_i^1, \star) +  d(\alpha_i^2, \star) + d(\alpha_i^3, \star) + d(\alpha_i^4, \star) = \\
         \left(U^5 + U^4 + 2 r_i U^3 - \frac{d_i}{4} \right) +  \left(U^5 + 2U^4 + 2 r_i U^3 - \frac{d_i}{4} \right) +  \\
         \left(U^5 + 3U^4 - 2 r_i U^3 + \frac{d_i}{4} \right) + \left(U^5 + 4U^4 - 2 r_i U^3 - \frac{d_i}{4} \right)     = 4 U^5 + 10 U^4.
    \end{align*}
 This concludes the proof. 
\end{proof}

\begin{lemma}
    \label{lem:diagonal}
    Let $B_1, B_2 \in \mathbf{B}$ be two blocks that are diagonally adjacent. 
    Denote by $R_{ij}$ the unique gadget that intersects both blocks. 
    Then for any $u \in B_1$ and $v \in B_2$ the shortest $xy$-monotone path from $u$ to $v$ has weight greater than $U^3$ if $r_i \neq c_j$.
    It has weight $4U + d_i \cdot b_j$ otherwise. 
\end{lemma}

\begin{proof}
    If $(B_1, B_2)$ are vertically adjacent then any shortest $xy$-monotone path from $u$ to $v$ consists of orange vertices plus exactly four grey vertices contained in $R_{ij}$.
 Since orange vertices have weight zero, it follows that the weight of this path is:
    \begin{align*}
         d(\alpha_i^1, \beta_j^1) +  d(\alpha_i^2, \beta_j^2) + d(\alpha_i^3, \beta_j^3) + d(\alpha_i^4, \beta^4_j) = \\
        \left| 2 (r_i - c_j) U^3 + \frac{d_i }{4} + U \right| 
    + \left| 2 (r_i - c_j) U^3 - \frac{d_i}{4} - U \right|
+  \\
\left| - 2 (r_i - c_j) U^3 + \frac{d_i}{4} - (-1)^{b_j} U \right| 
         +  \left| - 2 (r_i - c_j) U^3 - \frac{d_i}{4} + (-1)^{b_j} U \right| 
         \end{align*}
If $r_j \neq c_j$ this is at least $U^3$.
If $r_i = c_j$ and $b_j = 0$, this is then equal to:
$\frac{d_i}{4} + U + \frac{d_i}{4} + U - \frac{d_i}{4} + U - \frac{d_i}{4}  = 4 U$.
If $r_i = c_j$ and $b_j = 1$ this is equal to: $\frac{d_i}{4} + U + \frac{d_i}{4} + U + \frac{d_i}{4} + U + \frac{d_i}{4}  = 4 U + d_i$.    
\end{proof}

\begin{lemma}
    \label{lem:equivalentcost}
    For any instance of \textsc{Intermediary} with $n_r$ rows and $n_c$ columns,
    \begin{itemize}[noitemsep, nolistsep]
        \item If $\DTW(P, Q) \geq  |n_r + n_c| \cdot (4 U^5 + 10 U^4) + U^3$, then \textsc{Intermediary} outputs $\infty$.
        \item Else the output of \textsc{Intermediary} is equal to:
        \[
        DTW(P, Q) - |n_r - n_c| \cdot (4 U^5) +  4|n_r - n_c| U - \min\{ n_r, n_c \} \cdot 4 U.
        \]
    \end{itemize} 
\end{lemma}

\begin{proof}
    By Lemma~\ref{lem:nowhites} there exists a path $\pi^*$ in the rectangular graph induced by $P$ and $Q$ that realises $\DTW(P, Q)$ that intersects no white boundary vertices.
    It follows immediately that $\pi^*$ intersects a sequence of blocks $\mathbf{B}^* \subset \mathbf{B}$ where for every two consecutive blocks $B_1, B_2 \in \mathbf{B}^*$, $B_1$ and $B_2$ are either horizontally, vertically or diagonally adjacent. 
    The path $\pi^*$ may be partitioned into subpaths whose endpoints lie in consecutive blocks in $\mathbf{B}^*$.
    The weight of $\pi^*$ is equal to the weight of these subpaths. 

    For any consecutive blocks $B_1, B_2 \in \mathbf{B}^*$ that are horizontally or vertically adjacent, by Lemma~\ref{lem:horizontal},  the weight of any subpath of $\pi^*$ with its endpoints in $(B_1, B_2)$ is  $4 U^5 + 10 U^4$.

      For any consecutive blocks $B_1, B_2 \in \mathbf{B}^*$ that are diagonally adjacent (both intersecting the gadget $R_{ij}$), by Lemma~\ref{lem:diagonal},  the weight of any subpath of $\pi^*$ with its endpoints in $(B_1, B_2)$ is:
      \begin{itemize}[noitemsep]
          \item At least $U^3$ whenever $r_i \neq c_j$. 
          \item Equal to $4 U + d_i \cdot c_j$ otherwise. 
      \end{itemize}

    By Observation~\ref{obs:minimise}, the path $\pi^*$ takes as few white vertices as possible (prioritizing diagonals consisting of grey vertices whenever possible). Thus, it contains exactly $4 |n_r - n_c|$ white vertices. 
    This implies that there are exactly $4 |n_r - n_c|$ pairs of consecutive blocks $(B_1, B_2)$ that are horizontal or vertically adjacent, and $\min \{ n_r, n_c \}$ consecutive blocks $(B_1, B_2)$ that are diagonally adjacent. 
    By Lemma~\ref{lem:horizontal}, the subcurves of $\pi^*$ between $B_1$ and $B_2$ that are vertically or horizontally adjacent have a total weight of exactly 
    $|n_r - n_c| (4 U^5 + 10 U^4 )$. 
    We may apply the same argument to \textsc{Intermediary}, noting that the horizontal and vertical edges taken in \textsc{Intermediary} have a total weight of exactly $|n_r - n_c| U$. 

    We now consider two cases: 

    First, the case where \textsc{Intermediary} outputs $\infty$.
    In other words, the shortest path $(0, 0) \to (n_r - 1, n_c - 1)$ path in \textsc{Intermediary} is at least $|n_r - n_c| U  + \sqrt{U}$.
    This occurs if and only if there does not exist a path $\Pi$  in \textsc{Intermediary} from $(0, 0)$ to $(n_r - 1, n_c - 1)$ where for all diagonals from $(i, j)$ to $(i+1, j+1)$ in $\Pi$: $r_i = c_j$. 
    It follows by Lemma~\ref{lem:diagonal} that for the corresponding pairs of diagonal blocks $(B_1, B_2)$, the shortest path from any vertex $u \in B_1$ to any vertex $v \in B_2$ has weight at least $U^3$.
    We note that $\mathbf{B}^*$ must include at least one consecutive pair $(B_1, B_2)$ that is diagonally adjacent.  The subpath of $\pi^*$ between any such $B_1$ and $B_2$ has weight at least $U^3$ and the path $\pi^*$ has thus weight at least $|n_r - n_c| (4 U^5 + 10 U^4 ) + U^3$. 

    Second, the case where \textsc{Intermediary} outputs a finite value. 
    Consider each path $\Pi$  in \textsc{Intermediary} from $(0, 0)$ to $(n_r - 1, n_c - 1)$ where for all diagonals from $(i, j)$ to $(i+1, j+1)$ in $\Pi$: $r_i = c_j$. 
    Denote by $D = \{ (i, j) \}$ the set of diagonals taken by $\Pi$. 
    The cost of $\Pi$ is equal to $|n_r - n_c| U  + \sum_{(i, j) \in D} d_i \cdot b_j$.
    Since $\Pi$ is $xy$-monotone, there exists at least one path $\pi'$ in our rectangular grid graph where the corresponding block sequence $\mathbf{B}'$ contains $\min \{ n_c, n_r \}$ pairs of consecutive blocks $B_1, B_2 \in \mathbf{B}'$ that are diagonally adjacent where every such $B_1, B_2$ share a gadget $R_{ij}$ for a diagonal $(i, j) \in D$.

    Now consider the set of all paths $\pi'$, where the corresponding set of blocks $\mathbf{B}'$ contains $\min \{ n_c, n_r \}$ pairs of consecutive blocks $B_1, B_2 \in \mathbf{B}'$ that are diagonally adjacent where for all $(B_1, B_2)$ that share a gadget $R_{ij}$: $c_i = r_j$.
    Denote by $D'$ the set of all pairs $(i, j)$ for these gadgets $R_{ij}$.

    By our above analysis, the weight of $\pi'$ is equal to:
    \[|n_r - n_c| (4 U^5 + 10 U^4 )  + \sum\limits_{(i, j) \in D'} (4U + d_i \cdot b_j) =|n_r - n_c| (4 U^5 + 10 U^4 ) + 4 \min \{ n_r, n_c \} U +  \sum\limits_{(i, j) \in D'} d_i \cdot b_j. \]
    The path $\pi^*$ equals the path $\pi'$ with minimal weight and so the lemma follows. 
\end{proof}

\noindent
For any instance of \textsc{Intermediary}, we may compute $P$ and $Q$ in $O(n_c + n_r)$ time.
For each update in \textsc{Intermediary}, we only need to translate $4$ vertices in $Q$ in $O(1)$ time to maintain $(P, Q)$. 
By computing $DTW(P, Q)$ we may answer a query in \textsc{Intermediary} in $O(1)$ additional time. 
Thus:

\curvelowerbound*

Combining this with the lower bound (Theorem~\ref{thm:lowerboundInter}) from the next section gives:

\lowerboundDTW*
\section{Intermediary lower bound}
\label{sec:lowerbound}
In this section we prove the following theorem:

\lowerBoundIntermediary*

For example, Theorem~\ref{thm:lowerboundInter} implies that given polynomial preprocessing time, no data structure can have both the update and the query time significantly better than $O(n \cdot \sqrt{m})$.

\noindent
To this end, we recall the definition of Negative-$k$-Clique:

\kClique*

\paragraph{A switch in notation.}
To facilitate our proofs, we make a slight switch in notation.
For starters, we assume that in the \textsc{Intermediate} problem we have $n = n_r$ rows and $m = n_c$ columns.  
We refer to any edge in \textsc{Intermediate} its tail and its type (horizontal/vertical/diagonal).
For example the diagonal edge $(i,j)$ is the edge from $(i,j)$ to $(i+1,j+1)$.
When we say {\em shortest path} we always refer to a shortest path from $(0,0)$ to $(n-1, m-1)$.
We sometimes refer to $(0,0)$ as the top-left corner and to $(n-1, m-1)$ as the bottom-right corner.

We denote for any integer $A$ by $[A] := \set{0, \ldots, A-1}$. 
We denote for $A$ and $B$ with $B > A$ the integer intervals as $[A,B] = \set{A, A+1, \ldots, B}$, and $[A,B) = \set{A, A+1, \ldots, B-1}$.
Any positive integer $p\in[m]$ can be written as $p = \sum_{i\in [\ceil{\log_2{m}}]} b_i \cdot 2^i$, for some Booleans $b_i$.
We say that the binary number between bits $x$ and $y$ of $p$ is 
$\sum_{i\in [x,y]} b_i \cdot 2^{i-x}$.
For example $4 = 1\cdot 2^2 + 0\cdot 2^1 + 0\cdot 2^0$, and the number between bits $1$ and $2$ is $1\cdot 2^{2-1} + 0\cdot 2^{1-1} = 2$.

Finally, we use the notation $\Omega_k(\cdot), \Theta_k(\cdot),O_k(\cdot)$ to suppress factors depending only on $k$, the parameter of the Negative-$k$-Clique problem we are reducing from.
Throughout the proofs, one can think of $k$ as a sufficiently large constant.

\subsection{Some initial tools}
In our reduction from Negative-$k$-Clique to \Intermediary{} we construct instances of \Intermediary{} that satisfy the following additional restrictions.

\begin{assumption} \label{as:gridStructure}
Let $\gadgetSize{} = \Theta_k(N^2)$ be a parameter we specify later.

We assume that:
\begin{itemize}[noitemsep, nolistsep]
    \item $n\ge m$,
    \item both $n-1$ and $m-1$ are multiples of $\gadgetSize{}$,
    \item there exist constants $A_{-1}, A_0, A_1, A_2, A_3, A_4, A_5$ such that:
    \begin{itemize}[noitemsep, nolistsep]
        \item $ A_{-1} = 1$,
        \item $\floor{\sqrt{U}} \ge A_5$, and 
        \item $A_{i} = 100(n+N)^{10k} M \cdot A_{i-1} = N^{O(k)}$ for $i \in [6]$, where $M$ is defined as in Definition~\ref{def:kClique},
    \end{itemize}
    \item for any row $i$, it holds that $d_i < 2A_4$,
    \item there always exists an increasing sequence $s_0, s_1, \ldots, s_{m-2}$ such that $0\le s_0 < s_{m-2} < n-1$ and $r_{s_i} = c_{i}$ for all $i\in [m-1]$.
\end{itemize}
\end{assumption}
Intuitively, the gaps between the constants $A_i$ are so large that allow us to treat the constants independently.
Furthermore, there always exist a path that uses only diagonals of weight smaller than $\sqrt{U}$ and no horizontal edges (see Lemma~\ref{lem:countEdgeTypes}).
We later prove that the restrictions of \Cref{as:gridStructure} hold in the instances created by our reduction from Negative-$k$-Clique to \Intermediary{}.
Before that, we first prove some results related to instances having these restrictions.

We start with the following simple observation that it never helps to take a horizontal edge:

\begin{lemma} \label{lem:countEdgeTypes}
A shortest path from $(0,0)$ to $(n-1,m-1)$ uses no horizontal edges, exactly $n-m$ vertical edges, and exactly $m-1$ diagonal edges.
Furthermore, for each diagonal edge from $(i,j)$ to $(i+1,j+1)$ used from a shortest path, it holds that $r_i=c_j$ and therefore the edge's weight is less than $\sqrt{U}$.
Finally, the total weight of any shortest path is less than $(n-m)\cdot U + \sqrt{U}$.
\end{lemma}
\begin{proof}
Let $s_0, \ldots, s_{m-2}$ be the increasing sequence from \Cref{as:gridStructure}.
Consider the following path $P$ from $(0,0)$ to $(n-1,m-1)$. Whenever we are at vertex $(i,j)$, if $j=m-1$, we move vertically until we reach $(n-1,m-1)$. Else, if $i<s_j$ then we move (vertically) to $(i+1,j)$. Else, we move (diagonally) to $(i+1,j+1)$.

As $P$ does not use any horizontal edge, this means that it takes exactly $m-1$ diagonal edges. Furthermore, in every step it proceeds by one row, meaning it takes $n-1$ edges in total. Therefore $n-m$ of them are vertical edges, and the total cost of the path is $(n-m)\cdot U + \sum_{j=0}^{m-2} d_{s_j}\cdot b_j$.
As all diagonal edges used by the shortest path have weight less than $\sqrt{U}$, and by Assumption~\ref{as:gridStructure}, the total cost of the path is less than $(n-m) \cdot U + \sqrt{U}$.

On the other hand, the maximum amount of diagonal edges on any shortest path is $m-1$, meaning that any shortest path needs to take at least $n-m$ vertical edges. If it takes more vertical edges, or if it takes at least one horizontal edge, then the weight of the path is at least $(n-m+1)\cdot U$, meaning that it cannot be a shortest path. Therefore it takes no horizontal edge, exactly $n-m$ vertical edges, and  exactly $m-1$ diagonal edges.
If it takes a diagonal edge $(i,j)$ with $r_i \ne c_j$ then the cost is at least $(n-m) \cdot U + \sqrt{U}$ and the path is again not a shortest path.
\end{proof}

We now give results related to a certain structure we use in the main reduction.
Intuitively, we define a certain type of subgraph that we call a gadget, which can only be traversed diagonally in a shortest path. See Figure~\ref{fig:gammagadget} for an illustration. 

\begin{definition}
Let $\alpha_r, \alpha_c$ be non-negative integers, and $\Gamma$ be the rectangular subgraph induced by all vertices $(i,j)$ with $\alpha _r\le i \le \alpha_r + \gadgetSize{}$ and $\alpha_c \le j \le \alpha_c + \gadgetSize{}$.
We say $\Gamma$ is a \emph{$\gamma$-gadget} if:
\begin{itemize}
    \item for all $(i,j)$ contained in $\Gamma$ it holds that $r_i=c_j$ if and only if $i-\alpha_r = j-\alpha_c$,
    \item $d_{\alpha_r} = \gamma, b_{\alpha_c} = 1$,
    \item $d_{\alpha_r+\gadgetSize{}-1} = \constGadget{}-\gamma, b_{\alpha_r+\gadgetSize{}-1} = 1$,
\end{itemize}
We say that the $(i-\alpha_r, j-\alpha_c)$ diagonal edges are the \emph{main diagonal edges} of the $\gamma$-gadget. The \emph{cost} of a $\gamma$-gadget is equal to the sum of weights of its main diagonal edges $\sum_{p\in [g]} d_{\alpha_r+p}\cdot b_{\alpha_c+p}$.

For ease of notation, we sometimes say that $\Gamma$ is a \emph{gadget}, instead of a $\gamma$ gadget.
\end{definition}

\begin{figure}[h]
  \centering
  \includegraphics[width = 0.9\linewidth]{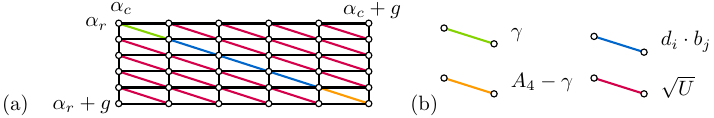}
  \caption{
(a) a $\gamma$-gadget. (b) we show for every diagonal edge their weight. 
    }
 \label{fig:gammagadget}
\end{figure}

From a high level view, the interesting property of a gadget is that a shortest path can either follow all its main diagonal edges, or not use any main diagonal edge of the gadget at all.
Indeed, entering the gadget comes at some high cost $\gamma$ for the main diagonal. 
Exiting the gadget through the last diagonal has an edge with a cost that includes $-\gamma$, to cancel out the earlier cost of $\gamma$. 
If $\gamma$ is large enough then a shortest path that paid the (large) cost $\gamma$, must reach the last diagonal edge of the gadget, in order to gain the $-\gamma$.

One more restriction of the instances created by our reduction from Negative-$k$-Clique to \Intermediary{} is the following:

\begin{assumption} \label{as:gadgets}
For any vertex $(i,j)$ with both $i$ and $j$ being multiples of $\gadgetSize{}$, we have an  $\left((n-i) \cdot \belowConstGadget{}\right)$-gadget with $(i,j)$ being its top left corner.
Furthermore, every row $i$  (where neither $i$ nor $i+1$ is a multiple of $\gadgetSize{}$) has weight $d_i$ at most $n\constEnabler{}$.
\end{assumption}

We note that a vertex may be contained in up to four gadgets.
For example, vertex $(\gadgetSize{}, \gadgetSize{})$ is shared by four gadgets, the ones with their top left corner being $(0, 0), (0, \gadgetSize{}), (\gadgetSize{}, 0), (\gadgetSize{}, \gadgetSize{})$.
On the other hand, every diagonal edge is contained in exactly one gadget.

Given \Cref{as:gadgets}, we obtain the following result:

\begin{lemma} \label{lem:gadgetWhole}
Assuming \Cref{as:gadgets}, any shortest path that uses a main diagonal edge of a gadget must use all main diagonal edges of this gadget.
\end{lemma}
\begin{proof}
Let $P$ be a shortest path.
If $P$ is at the top left corner $(i,j)$ of a gadget $\Gamma$ and follows a vertical edge, then it must continue vertically until it reaches the top left corner $(i+\gadgetSize{}, j)$ of another gadget, because by Lemma~\ref{lem:countEdgeTypes} $P$ cannot use any non-main diagonal edge.
Else, if $P$ follows a (main) diagonal edge of $\Gamma$, it either follows all main diagonal edges of $\Gamma$, in which case it reaches the top left corner $(i+\gadgetSize{},j+\gadgetSize{})$ of another gadget, or $P$ follows the first main diagonal edge $(i,j)$ of $\Gamma$ but not all of them.

Based on the above, and as $P$ starts at $(0,0)$, which is the top left corner of a gadget, there are two cases. Either the claim of the lemma directly holds, or there exists a first vertex $(\eta',\theta')$ which is the top left corner of a gadget $\Gamma$ such that $P$ follows the main diagonal edge $(\eta', \theta')$ of $\Gamma$, but not all the rest.

As $P$ does not use any horizontal edges, and we assumed it does not use all main diagonal edges of $\Gamma$, it cannot reach vertex $(\eta'+\gadgetSize{}, \theta'+\gadgetSize{})$.
Additionally, $P$ does not use any non-main diagonal edge, meaning it must use some main diagonal edge $(\mu-1, \theta'+\gadgetSize{}-1)$ of another gadget, where $\mu$ is a multiple of $\gadgetSize{}$.
We show that we can modify the subpath of $P$ between vertices $(\eta'+1, \theta'+1)$ and $(\mu, \theta'+\gadgetSize{})$ while reducing the cost, thus contradicting the fact that $P$ is a shortest path.

Notice that in this subpath $P$ uses $\mu-\eta'-\gadgetSize{}$ vertical edges (cost $(\mu-\eta'-\gadgetSize{})U$), and also pays $\constGadget{} - (n-\mu-\gadgetSize) \cdot \belowConstGadget{}$ for the final diagonal edge.
We instead use all the main diagonal edges of $\Gamma$, thus reaching $(\eta'+\gadgetSize{}, \theta'+\gadgetSize{})$, and then use vertical edges to reach $(\mu, \theta'+\gadgetSize{})$.
The cost is $(\mu-\eta'-\gadgetSize{})U$ for the vertical edges, plus $\constGadget{} - (n-\eta') \cdot \belowConstGadget{}$ for the final diagonal edge, plus the cost of the rest of the diagonal edges.
By assumption, the cost of each of the rest of the diagonal edges is less than $n \constEnabler{}$.
We conclude that the new cost is improved by at least $\belowConstGadget{} - n^2\constEnabler{} > 0$.
\end{proof}

\subsection{Reduction} \label{ssec:reduction}
We now describe the reduction from Negative-$k$-Clique to \Intermediary{}.

\paragraph{Weight function}
We assume that the input graph $G_0=(V=V_0\cup V_1\cup \ldots \cup V_{k-1}, E, w)$ has $N$ nodes, where $N$ is a multiple of $k$, and that $G$ is a complete $k$-partite graph.
Each of the parts $V_i$ contains exactly $N/k$ nodes.
We identify the $N$ nodes with integers in $[N]$, such that $V_i = [i\cdot N/k, (i+1)\cdot N/k)]$.
Notice that any $k$ clique must have exactly one node in each $V_i$.

The function $w$ encodes the weight of an edge, that is for $\set{u,v} \in E$ we have that the weight of the edge connecting $u,v$ is $w(u,v)$.
Recall that $M = N^{O(k)}$ is equal to $\sum_{\set{u,v} \in E} |w(u,v)|$.
We introduce two new auxiliary nodes $\alpha=N, \beta=N+1$, and we extend $w$ so that $w(\alpha, \beta)=w(\alpha,u)=w(\beta,u)=2M$, for each $u\in V$. 
Similarly we extend $w(u,u)=0$ and $w(u,v)=2M$ for $u,v$ in the same part $V_i$ and $u\ne v$.
This ensures that if $S$ is a node set of size at least $2$ containing $\alpha$ or $\beta$ or two nodes from the same part $V_i$, then $\sum_{u,v\in S, u<v} w(u,v)$ is too large (at least $M$).

Furthermore, we extend $w$ to take sets as arguments.
For two node sets $S,T$ we define $w(S,T)$ to be the sum of $w(i,j)$ over all unordered pairs $\set{i,j}$ with one endpoint in $S$ and the other in $T$.
Notice that this is not the same as $\sum_{u\in S, v\in T} w(u,v)$; for example $w(S,S)$ does not double count the weight of each pair.

\paragraph{Splitting the parts of $V$}
In what follows, we partition the set $\{ V_i \mid i \in [k] \}$ into four disjoint sets $\mathcal{V}_1, \mathcal{V}_2, \mathcal{V}_3, \mathcal{V}_4$.
For each $i\in \{1,2,3,4\}$, we refer to each  $V_j \in \mathcal{V}_i$ as a part of $V$ that contains $\frac{N}{k}$ nodes. 
We may select for all parts $V_j \in \mathcal{V}_i$ one such node to create a \emph{sequence} of vertices $(v_0,v_1, \ldots v_{|\mathcal{V}_i|-1})$. 
There are then  $(\frac{N}{k})^{|\mathcal{V}_i|}$ sequences that can be generated by picking vertices this way. 
We describe a very non-straightforward way to iterate over these sequences when $i=1$ or $i=3$.
This helps us encode interesting information in our gadgets in Section~\ref{ssec:implementGadgets}.

More formally, we use three positive parameters $\rho_1, \rho_2, \rho_3$, with $\rho_1+\rho_2+\rho_3<1$, that we fix later.
Based on these parameters, we split the parts of $V$ into four disjoint sets, the first containing $\floor{\rho_1 k}$ parts, the second $\floor{\rho_2 k}$ parts, the third $\floor{\rho_3 k}$ parts, and the last containing the rest of the parts (at most $(1-\rho_1-\rho_2-\rho_3)k+3$ parts).
More formally, let
\begin{itemize}
    \item $\mathcal{V}_1 = \set{V_0, \ldots, V_{\floor{\rho_1 k}-1}}$,
    \item $\mathcal{V}_2 = \set{V_{\floor{\rho_1 k}}, \ldots, V_{\floor{\rho_1 k}+\floor{\rho_2 k}-1}}$,
    \item $\mathcal{V}_3 = \set{V_{\floor{\rho_1 k}+\floor{\rho_2 k}}, \ldots, V_{\floor{\rho_1 k}+\floor{\rho_2 k} + \floor{\rho_3 k}-1}}$,
    \item $\mathcal{V}_4 = \set{V_{\floor{\rho_1 k}+\floor{\rho_2 k}+\floor{\rho_3 k}}, \ldots, V_{k-1}}$.
\end{itemize}

\noindent
For $i\in \set{2,4}$ let $\tau_i=(\frac{N}{k})^{|\mathcal{V}_i|}$. We denote by $U_i$ the set of
all $\tau_i$ sequences (containing exactly one node from each part in $\mathcal{V}_i$).
We order $U_i$ in arbitrary order, and use the notation $U_i(j)$ to refer to the $j$-th sequence in $U_i$.
Iterating over all $U_i(j), j\in [\tau_i]$, we generate all different sequences (each containing exactly one node from each part in $\mathcal{V}_i$).

For $i\in \set{1,3}$, our way of iterating over all different sets containing exactly one node from each part in $\mathcal{V}_i$ is more technical.
Let $\overline{N} = 2^{1+\ceil{\log_2{N}}}$, that is $\overline{N}$ is twice the smallest power of $2$ that is at least as large as $N$.
Given a sequence of $x$ nodes $u_0, u_1, \ldots, u_{x-1}$, we let $f(u_0, u_1, \ldots, u_{x-1}) = \sum_{i=0}^{x-1} u_i \overline{N}^i$.
Notice that, as $\overline{N}$ is larger than $N$, for any $y\in [x]$, we can retrieve $u_y$, given $f(u_0, u_1, \ldots, u_{x-1})$.
Furthermore, since $\overline{N}$ is a power of $2$, it suffices to read the binary number $z$ between bits $\log_2{\overline{N}}\cdot y$ and $\log_2{\overline{N}}\cdot (y+1)-1$, to retrieve $u_y$.
In fact, as $u_y$ consists of $\ceil{\log_2{N}}$ bits, the topmost bit of $z$ is always zero.
We need to ensure this technicality for reasons that will become apparent in Section~\ref{ssec:implementGadgets}.

Seen in the reverse order, given a non-negative integer $p \in [\overline{N}^{|\mathcal{V}_1|}]$, we define a sequence $U_1(p)=(u_0, u_1, \ldots, u_{|\mathcal{V}_1|-1})$.
Let $z_y$ be the binary number between bits $\log_2{\overline{N}}\cdot y$ and $\log_2{\overline{N}}\cdot (y+1)-1$ of $p$, for $y\in [|\mathcal{V}_1|]$.
Then we define $u_y=z_y$ if $z_y \in V_y$, and $u_y=\alpha$ otherwise.

Similarly, given a non-negative $p \in [\overline{N}^{|\mathcal{V}_3|}]$, we define a sequence $U_3(p)=(v_0, v_1, \ldots, v_{|\mathcal{V}_3|-1})$.
Let $z'_y$ be the binary number between bits $\log_2{\overline{N}}\cdot y$ and $\log_2{\overline{N}}\cdot (y+1)-1$ of $p$, for $y\in [|\mathcal{V}_3|]$.
Then we define $v_y=z'_y$ if $z'_y \in V_{|\mathcal{V}_1|+|\mathcal{V}_2|+y}$, and $v_y=\beta$ otherwise.

For $i\in \set{1,3}$, let $\tau_i = \overline{N}^{|\mathcal{V}_i|}$, and notice that iterating over all $U_i(j), j\in [\tau_i]$, we generate all different sequences containing exactly one node from each part in $\mathcal{V}_i$ (along with some ``garbage'' sequences that contain the node $\alpha$ or $\beta$).
Furthermore, $\tau_i = O_k(N^{|\mathcal{V}_i|})$, so intuitively the redundancy we introduce is small.

We use $U_i(j) \circ U_{i'}(j')$ to refer to the concatenation of the two sequences $U_i(j)$ and $U_{i'}(j')$, for $i,i'\in [4], j\in [\tau_{i}], j'\in [\tau_{i'}]$.

\paragraph{Applying \textsc{Intermediary}}

\begin{figure}[t]
  \centering
  \includegraphics[]{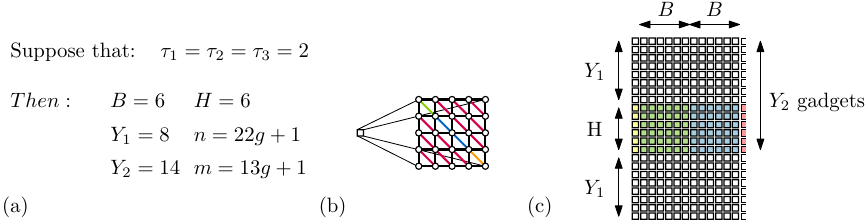}
  \caption{
(a) We choose $\tau_1 = \tau_2 = \tau_3 = 2$, which determine $B,Y_1,Y_2,H,n,m$.
(b) By a small square, we denote our gadget that is a rectangular graph of $g \times g$. 
(c) We illustrate our reduction: creating $\tau_2 = 2$ blocks (green + blue). 
    }
 \label{fig:reductionInter}
\end{figure}

For a given instance of Negative-$k$-Clique, we define the following variables:
\begin{itemize}[noitemsep]
    \item $\gadgetSize{} = 12+ \ceil{\log_2{\constStayBigDiagonal{}}} + N + (|\mathcal{V}_1|+|\mathcal{V}_2|+|\mathcal{V}_3|)|\mathcal{V}_4|(N+2) + 2\overline{N}(|\mathcal{V}_1|+|\mathcal{V}_2|)|\mathcal{V}_3|(N+2)$
	\item $\block{} = \blockFormula$,
	\item $\mainHeight{} = \mainHeightFormula{}$,
	\item $\mainRow{} = \mainRowFormula{}$,
	\item $\mainRowEnd{} = \mainRowEndFormula{}$,
	\item $n= (\rows{}) \gadgetSize{}+1$,
	\item $m = (\cols{}) \gadgetSize{}+1$.
\end{itemize}

Until \Cref{ssec:implementGadgets}, we only need that $g=O_k(N^2)$.
We subsequently create an instance of \Intermediary{} by creating a rectangular graph that consists of $O(n/g) \times O(m/g)$ gadgets (see Figure~\ref{fig:reductionInter}).
The rows of our graph have a head and a tail of $Y_1$ gadgets, and a center of $H$ gadgets. 
The columns of our graph have a head of $1$ gadget, followed by a center of $\tau_2 B$ gadgets. 
This splits the center of the rectangular graphs (defined by the centers of the rows and columns) into $\tau_2 B$ blocks of $H \times B$ gadgets. 
The high-level idea of our reduction, is that we construct our gadgets in such a way that the shortest path in \Intermediary{} crosses exactly one block, and follows
a path with certain properties
within that block.
Each combination of a block and
path with the aforementioned properties
corresponds to a $k$-Clique in our problem instance.

\paragraph{Showing that we may apply \Cref{as:gridStructure}.}
Notice that $\gadgetSize = \Theta_k(N^2), n\ge m$ and $n-1,m-1$ are multiples of $\gadgetSize{}$.
Recall the constant $A_i$ from Assumption~\ref{as:gridStructure}.
In \Cref{ssec:implementGadgets} we specify the weight of each row and it is straightforward to verify that $max_{i\in [n]} d_i < 2\constGadget{}$.
Let $\%$ be the modulo operator.
We set the identifier of row $i$ to $r_i=i\% \gadgetSize{}$ and similarly the identifier of column $j$ to $c_j=j\%\gadgetSize{}$.
Therefore the increasing sequence of Assumption~\ref{as:gridStructure} exists, with $s_i=i$ being a witness.

Finally, we let $U = 1+n\cdot m\cdot (\max_i{d_i}\cdot \max_i{r_i} \cdot \max_j{c_j} \cdot A_5)^2 = N^{O(k)}$, which ensures both the requirement from the statement of \Intermediary{} that $U > n\cdot m\cdot (\max_i{d_i}\cdot \max_i{r_i} \cdot \max_j{c_j})^2$, and that $\floor{\sqrt{U}} \ge A_5$.
We conclude that \Cref{as:gridStructure} indeed holds.

\paragraph{Showing that we may apply \Cref{as:gadgets}}
We straightforwardly ensure \Cref{as:gadgets}: We make each vertex $(i,j)$ with both $i,j$ being multiples of $\gadgetSize{}$ the top left corner of a $(n - i)\belowConstGadget{}$ gadget.
Furthermore, for any $y$ such that neither $y$ nor $y+1$ is a multiple of $\gadgetSize{}$, we set the weight of row $y$ to be at most $n\cdot \constEnabler{}$ (see \Cref{ssec:implementGadgets} for a specification of the row weights).

\paragraph{} For ease of notation, we refer to the gadget with top left corner $(i \gadgetSize{}, j\gadgetSize{})$ as the $(i,j)$ gadget.
We say that we {\emph{use the $(x,y)$ gadget}} to denote that we move from $(x\gadgetSize{}, y\gadgetSize{})$ to $((x+1)\gadgetSize{}, (y+1)\gadgetSize{})$ using all the main diagonal edges of the gadget.
As defined earlier, the cost of the $(x,y)$ gadget is $\sum_{i\in [\gadgetSize{}]} d_{x\gadgetSize{}+i} b_{y\gadgetSize{}+i}$, that is the cost a path pays to use this gadget.

Notice that Lemma~\ref{lem:gadgetWhole} applies, and it hints that we can view gadgets as single diagonal edges.
In what follows we heavily use Lemma~\ref{lem:gadgetWhole}, even without explicitly stating it.

Our reduction works in $\tau_4$ epochs, each one corresponding to a different $U_4(s)$.
Each of these epochs is further divided into $O(\tau_1)$ phases, each one corresponding to a different $U_1(\cdot)$.

When we are at the $s$ epoch and the $t$ phase, we say we are at phase $(s,t)$.
For every phase $(s,t)$ we always have $s\in [\tau_4], t\in [\tau_1]$ and $U_1(\tau_1-t-1)\not \ni \alpha$.

We now describe a certain type of path from $(0,0)$ to $(n-1,m-1)$ that we call a \emph{restricted} path (see \Cref{fig:restrictedPath}).
We later show that a shortest path needs to be a restricted path.

\begin{figure}[t]
  \centering
  \includegraphics[width=250px, height=250px]{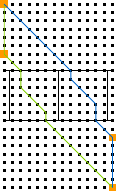}
  \caption{
  The graph we get when $\tau_1=\tau_2=\tau_3=2$, and two restricted paths at phase $(s,t=0)$.
  Each node not in the last row or the last column represents a gadget.
  A restricted path can only use the $(t=0,0)$ or the $(\block{}+t=\block{},0)$ gadget, out of all the $(y,0)$ gadgets.
  Then it follows diagonals until it reaches height $\mainRow{}\gadgetSize{}$.
  Subsequently it moves vertically, then uses $2t+3$ gadgets, and then moves vertically again until it reaches height $\mainRowEnd{}\gadgetSize{}$.
  Finally, it moves diagonally until it reaches the last column, and then vertically.
  \\
  To ensure that a shortest path is actually a restricted path, we make the $(0,0)$, $(6,0)$, $(15,12)$, $(21,12)$ gadgets much cheaper than the rest of the $(y,0)$, $(y,12)$ gadgets. 
  Furthermore, we make the $(15,12)$ gadget cheaper than the $(21,12)$ gadget, which ``matches" the $(15,12)$ gadget with the $(0,0)$ gadget, but not with the $(6,0)$ gadget (as $(6,0)$ is too low to reach $(15,12)$).
  Additionally, for $x\in [1,12)$ the $(y,x)$ gadgets $y \not \in [\mainRow{}, \mainRowEnd{})$ are much cheaper than the ones with $y \in [\mainRow{}, \mainRowEnd{})$.
  This ensures that when at vertex $(y\gadgetSize{}, x\gadgetSize{})$, $y \not \in [\mainRow{}, \mainRowEnd{})$, a shortest path always uses the $(y,x)$ gadget.
    }
 \label{fig:restrictedPath}
\end{figure}

\begin{definition}
For some $j\in [\tau_2], l\in [\tau_3]$ during phase $(s,t)$, a path from $(0,0)$ to $(n-1,m-1)$ is said to be a \emph{$(j,l)$-restricted path}, or simply a \emph{restricted} path, if it has the following form (Figure~\ref{fig:restrictedPath})

\begin{enumerate}
    \item We move vertically from the top left corner of the $(0,0)$ gadget to the top left corner of the $(j\block{}+t,0)$ gadget.
    \item We use the $(j\block{}+t,0)$ gadget.
    \item We move diagonally to the top left corner of the $(\mainRow{}, \mainRow{}-j\block{}-t) = (\mainRow{}, (\tau_2-j-1)\block{}+\tau_1-t)$ gadget.
    \item We move vertically to the top left corner of the $(\mainRow{}+\tau_1-t-1+l, (\tau_2-j-1)\block{}+\tau_1-t)$ gadget.
    \item We move diagonally to the top left corner of the $(\mainRow{}+\tau_1+l-1, (\tau_2-j-1)\block{}+\tau_1)$ gadget.
    \item We use the $(\mainRow{}+\tau_1+l-1, (\tau_2-j-1)\block{}+\tau_1)$ gadget.
    \item We use the $(\mainRow{}+\tau_1+l, (\tau_2-j-1)\block{}+\tau_1+1)$ gadget.
    \item We move diagonally to the top left corner of the $(\mainRow{}+\tau_1+t+l+2, (\tau_2-j-1)\block{}+\tau_1+t+3)$ gadget.
    \item We move vertically to the top left corner of the $(\mainRow{}+2\tau_1+\tau_3, (\tau_2-j-1)\block{}+\tau_1+t+3) = (\mainRowEnd{},(\tau_2-j-1)\block{}+\tau_1+t+3))$ gadget.
    \item We move diagonally to the top left corner of the $(\mainRowEnd{}+j\block{}+\tau_1-t-1, (\tau_2-1)\block{}+2\tau_1+2) = (\mainRowEnd{}+j\block{}+\tau_1-t-1, \tau_2\block{})$ gadget. 
    \item We use the $(\mainRowEnd{}+j\block{}+\tau_1-t-1, \tau_2\block{})$ gadget. 
    \item We move vertically from $((\mainRowEnd{}+j\block{}+\tau_1-t)\gadgetSize{},m-1)$ to $(n-1,m-1)$.
\end{enumerate}

\end{definition}

For our reduction, we first describe the desired costs of the gadgets at phase $(s,t)$, and show the lower bound.
In \Cref{ssec:implementGadgets} we show how to implement the gadgets, by the weights $d_i$ of rows, as well as the activations $b_j$ of columns.

In what follows, one can think of $h(i,j,l)$ as being equal to $w(U_1(i)\circ U_2(j), U_3(l))$.
For reasons related to the actual implementation of of our gadget costs, $h(\cdot, \cdot, \cdot)$ is a function with a very technical definition, that we specify in \Cref{ssec:implementGadgets}.

\begin{definition} \label{def:desiredCosts}
Let $h$ be a function satisfying the following constraints: $h(i,j,l)$, with $i\in [\tau_1], j\in [\tau_2], l\in [\tau_3]$, is equal to $w(U_1(i)\circ U_2(j), U_3(l))$ when $U_1(i)\not \ni \alpha$ and $U_3(l)\not \ni \beta$.
Furthermore, for any $i,j$ such that $U_1(i)\not \ni \alpha$, there exists an $l_{i,j}\in [\tau_3]$ such that $U_3(l_{i,j}) \not \ni \beta$ and $h(i,j,l_{i,j}) \le h(i,j,l)$ for any $l \in [\tau_3]$.
Finally $h(i,j,l) \le 2N^2M$ for any $i,j,l$.

At phase $(s,t)$, with $U_1(\tau_1-t-1)\not \ni \alpha$, we say that the gadgets have the desired costs if the following hold:
\begin{itemize}
    \item The $(j\block{}+t,0)$ gadgets, $j\in [\tau_2]$, have cost $\constGadget{} + (\tau_2-j)\constStayBigDiagonal{} + 2\gadgetSize{}M +
    \\w(U_1(\tau_1-t-1)\circ U_2(\tau_2-j-1), U_1(\tau_1-t-1)\circ U_2(\tau_2-j-1)\circ U_4(s))$.
    \item The $(\mainRowEnd{} + j\block{}+\tau_1-t-1, \tau_2\block{})$ gadgets, $j\in [\tau_2]$, have cost $\constGadget{} + j\cdot \constStayBigDiagonal{}$.

    \item The $(y_1,0)$ gadgets and the $(y_2, \tau_2\block{})$ gadgets, with $y_1<\mainRowEnd{}, y_1 \not \in \set{j\block{}+t \mid j\in [\tau_2]}, y_2\ge \mainRow{}, y_2\not \in \set{\mainRowEnd{} + j\block{}+\tau_1-t-1 \mid j\in [\tau_2]}$, have cost at least $\constGadget{} + \constEnabler{}$.
    
    \item The $(y,x)$ gadgets, for $y<\mainRow{}$ or $y\ge \mainRowEnd{}$, and $0<x<\tau_2\block{}$, have cost $\constGadget$.

    \item The $(\mainRow{}+i+l, j\block+i+1)$ gadgets, with $i\in [\tau_1-1], j\in [\tau_2], l\in [\tau_3]$, have cost $\constGadget{} + (\mainHeight{}-l)\constStaySmallDiagonal + 4\gadgetSize{}M + h(i,j,l) - h(i+1,j,l)$.
    \item The $(\mainRow{}+\tau_1-1+l, j\block+\tau_1)$ gadgets, with $j\in [\tau_2], l\in [\tau_3]$, have cost $\constGadget{} + (\mainHeight{}-l)\constStaySmallDiagonal + 2\gadgetSize{}M + h(\tau_1-1,j,l)$.
    \item The $(\mainRow{}+y, j\block{}+i+1)$ gadgets, $i\in [\tau_1], j\in [\tau_2], y\in [\mainHeight{}] \setminus [i,i+\tau_3)$, have cost at least $\constGadget{} + (\mainHeight{}+i-y)\constStaySmallDiagonal$.
    \item The $(\mainRow{}+\tau_1+l, j\block{}+\tau_1+1)$ gadgets, with $j\in [\tau_2], l\in [\tau_3]$, have cost $\constGadget{} + 2\gadgetSize{}M + w(U_3(l),U_3(l)\circ U_4(s))$.
    \item The $(\mainRow{}+y, j\block{}+\tau_1+1)$ gadgets, with $y\in [\tau_1]\cup [\tau_1+\tau_3,\mainHeight{})$, and $j\in [\tau_2]$, have cost at least $\constGadget{} + \constEnabler{}$. 
    \item The $(\mainRow{}+y, j\block{}+\tau_1+i+2)$ gadgets, with $i\in [\tau_1], j\in [\tau_2], y\in [\mainHeight{}]$, have cost $\constGadget{} + (\mainHeight{}-\tau_1-1+y-i) \constStaySmallDiagonal{}$.
    \item The $(\mainRow{}+y, j\block{}+2\tau_1+2)$ gadgets, with $j\in [\tau_2], y\in [\mainHeight{}]$, have cost at least $\constGadget{} + \constEnabler{}$.
\end{itemize}
\end{definition}

We now show what the cost of a restricted path at phase $(s,t)$ is:

\begin{lemma} \label{lem:upperBoundCost}
For any $j\in [\tau_2], l\in [\tau_3]$ during phase $(s,t)$, the \emph{$(j,l)$-restricted path} has cost 
\begin{align*}
(n-m)U +(\cols{})\constGadget + \tau_2\constStayBigDiagonal{} + 2(t+1)\mainHeight{}\cdot \constStaySmallDiagonal{} + (4t+6)M + h(\tau_1-t-1,\tau_2-j-1, l) +\\
w(U_1(\tau_1-t-1)\circ U_2(\tau_2-j-1), U_1(\tau_1-t-1)\circ U_2(\tau_2-j-1)\circ U_4(s)) + w(U_3(l),U_3(l)\circ U_4(s))
\end{align*}
\end{lemma}
\begin{proof}
We analyze the cost of the \emph{$(j,l)$-restricted path} step by step.

\begin{enumerate}
    \item We move vertically from the top left corner of the $(0,0)$ gadget to the top left corner of the $(j\block{}+t,0)$ gadget.\\
    Cost $= (j\block{}+t)\cdot \gadgetSize U$.
    \item We use the $(j\block{}+t,0)$ gadget.\\
    Cost $=\constGadget{} + (\tau_2-j)\cdot \constStayBigDiagonal{} + 2gM +
    w(U_1(\tau_1-t-1)\circ U_2(\tau_2-j-1), U_1(\tau_1-t-1)\circ U_2(\tau_2-j-1)\circ U_4(s))$.
    \item We move diagonally to the top left corner of the $(\mainRow{}, \mainRow{}-j\block{}-t) = (\mainRow{}, (\tau_2-j-1)\block{}+\tau_1-t)$ gadget.\\
    Cost $=(\mainRow{}-j\block{}-t-1)\constGadget$.
    \item We move vertically to the top left corner of the $(\mainRow{}+\tau_1-t-1+l, (\tau_2-j-1)\block{}+\tau_1-t)$ gadget.\\
    Cost $= (\tau_1-t-1+l)\gadgetSize{} U$.
    \item We move diagonally to the top left corner of the $(\mainRow{}+\tau_1+l-1, (\tau_2-j-1)\block{}+\tau_1)$ gadget.\\
    Cost telescoping to $t(\constGadget{}+(\mainHeight{}-l) \cdot \constStaySmallDiagonal + 4\gadgetSize{}M)  + h(\tau_1-t-1,\tau_2-j-1, l) - h(\tau_1-1,\tau_2-j-1,l)$.
    \item We use the $(\mainRow{}+\tau_1+l-1, (\tau_2-j-1)\block{}+\tau_1)$ gadget.\\
    Cost $= \constGadget{}+(\mainHeight{}-l) \cdot \constStaySmallDiagonal + 2\gadgetSize{}M  + h(\tau_1-1,\tau_2-j-1, l)$.
    \item We use the $(\mainRow{}+\tau_1+l, (\tau_2-j-1)\block{}+\tau_1+1)$ gadget.\\
    Cost $= \constGadget{} + 2\gadgetSize{}M + w(U_3(l),U_3(l)\circ U_4(s))$.
    \item We move diagonally to the top left corner of the $(\mainRow{}+\tau_1+t+l+2, (\tau_2-j-1)\block{}+\tau_1+t+3)$ gadget.\\
    Cost $=(t+1)(\constGadget{}+(\mainHeight{}+l)\cdot \constStaySmallDiagonal{})$.
    \item We move vertically to the top left corner of the $(\mainRow{}+2\tau_1+\tau_3, (\tau_2-j-1)\block{}+\tau_1+t+3) = (\mainRowEnd{},(\tau_2-j-1)\block{}+\tau_1+t+3))$ gadget.\\
    Cost $=(\tau_1+\tau_3-t-l-2)\gadgetSize \cdot U$.
    \item We move diagonally to the top left corner of the $(\mainRowEnd{}+j\block{}+\tau_1-t-1, (\tau_2-1)\block{}+2\tau_1+2) = (\mainRowEnd{}+j\block{}+\tau_1-t-1, \tau_2\block{})$ gadget.\\ 
    Cost $=(j\block{}+\tau_1-t-1)\constGadget$.
    \item We use the $(\mainRowEnd{}+j\block{}+\tau_1-t-1, \tau_2\block{})$ gadget.\\ 
    Cost $=\constGadget{}+j\cdot \constStayBigDiagonal{}$.
    \item We move vertically from $((\mainRowEnd{}+j\block{}+\tau_1-t)\gadgetSize{},m-1)$ to $(n-1,m-1)$.\\
    Cost $= (\mainRow{}-j\block{}-\tau_1+t) \gadgetSize{} U$.
\end{enumerate}

Summing up all the costs proves the lemma.
\end{proof}

\begin{corollary} \label{cor:costSP}
During phase $(s,t)$, let $C_{s,t}$ be the weight of the minimum weight $k$-Clique that contains the nodes in $U_1(\tau_1-t-1) \circ U_4(s)$.
For $j\in [\tau_2], l\in [\tau_3]$, the minimum cost of any $(j,l)$-restricted path is $(n-m)U +(\cols{})\constGadget + \tau_2\constStayBigDiagonal{} + 2(t+1)\mainHeight{}\cdot \constStaySmallDiagonal{} + (4t+6)M + C_{s,t} - w(U_4(s), U_4(s))$.
\end{corollary}
\begin{proof}
For a phase $(s,t)$, we always assume that $s\in [\tau_4], t\in [\tau_1], U_1(\tau_1-t-1)\not\ni \alpha$.
When $j,l$ are such that $U_3(l)\not\ni \beta$, we have that $h(\tau_1-t-1,\tau_2-j-1, l) = w(U_1(\tau_1-t-1)\circ U_2(\tau_2-j-1), U_3(l))$.
Therefore the cost of any $(j,l)$-restricted path, with $U_3(l)\not \ni \beta$, is
\begin{align*}
&(n-m)U +(\cols{})\constGadget + \tau_2\constStayBigDiagonal{} + 2(t+1)\mainHeight{}\cdot \constStaySmallDiagonal{} + (4t+6)M + \\
&\quad w(U_1(\tau_1-t-1)\circ U_2(\tau_2-j-1), U_3(l)) + w(U_3(l),U_3(l)\circ U_4(s)) + \\
&\quad w(U_1(\tau_1-t-1)\circ U_2(\tau_2-j-1), U_1(\tau_1-t-1)\circ U_2(\tau_2-j-1)\circ U_4(s)) = \\
&(n-m)U +(\cols{})\constGadget + \tau_2\constStayBigDiagonal{} + 2(t+1)\mainHeight{}\cdot \constStaySmallDiagonal{} + (4t+6)M + \\ &\quad w(U_1(\tau_1-t-1)\circ U_2(\tau_2-j-1)\circ U_3(l) \circ U_4(s), U_1(\tau_1-t-1)\circ U_2(\tau_2-j-1)\circ U_3(l) \circ U_4(s)) -\\
&\quad w(U_4(s), U_4(s))
\end{align*}

By definition of $C_{s,t}$, the minimum cost of any $(j,l)$-restricted path, with $U_3(l)\not \ni \beta$ is therefore $(n-m)U +(\cols{})\constGadget + \tau_2\constStayBigDiagonal{} + 2(t+1)\mainHeight{}\cdot \constStaySmallDiagonal{} + (4t+6)M + C_{s,t} - w(U_4(s), U_4(s))$.

We now need to argue that when $j,l$ are such that $U_3(l)\ni \beta$, the $(j,l)$-restricted path has larger cost.
In these cases, by \Cref{def:desiredCosts} there exists an $l_{\tau_1-t-1,\tau_2-j-1}$ such that $U_3(l_{\tau_1-t-1,\tau_2-j-1})\not \ni \beta$ and $h(\tau_1-t-1,\tau_2-j-1,l_{\tau_1-t-1,\tau_2-j-1}) \le h(\tau_1-t-1,\tau_2-j-1,l)$.

At the same time we have that $w(U_3(l_{\tau_1-t-1,\tau_2-j-1}),U_3(l_{\tau_1-t-1,\tau_2-j-1})\circ U_4(s)) \le M$ by definition of $M$ and the fact that $U_3(l_{\tau_1-t-1,\tau_2-j-1})\not \ni \beta$, while $w(U_3(l),U_3(l)\circ U_4(s)) \ge 2M$ by the fact that $U_3(l_{\tau_1-t-1,\tau_2-j-1}) \ni \beta$.

We conclude that the $(j,l)$-restricted path's cost is larger than the $(j,l_{\tau_1-t-1,\tau_2-j-1})$-restricted path's cost.
\end{proof}

We now show that during phase $(s,t)$, the shortest path is a restricted path.

\begin{lemma} \label{lem:structureSP}
At phase $(s,t)$ any shortest path is a restricted path.
\end{lemma}
\begin{proof}
The proof proceeds in four steps. Each step relates to some constants used in our construction:

\paragraph{$U$ and $\constGadget{}$:}
By Lemmas~\ref{lem:countEdgeTypes} and \ref{lem:gadgetWhole}, a shortest path uses exactly $\tau_2\block{}+1$ gadgets, exactly $n-m$ vertical edges, and no horizontal edges.
As every gadget has a cost of at least $\constGadget{}$, and every vertical edge costs $U$, the cost of any shortest path is at least $(n-m)\cdot U + (\tau_2\block{}+1)\constGadget{}$.

\paragraph{$\constEnabler{}$:}
Due to Lemma~\ref{lem:upperBoundCost} and Assumption~\ref{as:gridStructure}, the cost of a shortest path must be smaller than $(n-m)\cdot U + (\tau_2\block{}+1)\constGadget{} + \constEnabler{}$.

If $y_1\ge \mainRowEnd{}$ then vertex $(n-1,m-1)$ is unreachable from the top left corner of the $(y_1,0)$ gadgets.
Similarly if $y_2 < \mainRow{}$ then the top left corner of the $(y_2,\tau_2\block{})$ gadget is unreachable from $(0,0)$.
Out of the rest of the $(y_1,0)$ gadgets and $(y_2,\tau_2\block{})$ gadgets (for $y_1 < \mainRowEnd{}, y_2 \ge \mainRow{}$), the only ones with cost smaller than $\constGadget{} + \constEnabler{}$ are the $(j\block{}+t,0)$ gadgets and the $(\mainRowEnd{}+j\block{}+\tau_1-t-1, \tau_2\block{})$ gadgets.

Therefore a shortest path must first move vertically to the top left corner of a $(j_1\block{}+t,0)$ gadget and use it.
Furthermore, at some point it must use some $(\mainRowEnd{}+j_2\block{}+\tau_1-t-1, \tau_2\block{})$ gadget, and from then on it can only move vertically to $(n-1,m-1)$.

We also prove that $j_2 \ge j_1$.
The reason is that the maximum number of $(y,x)$ gadgets with $y<\mainRow{}$ we can use is $\mainRow{}-j_1\block{}-t$, and the maximum number of $(y,x)$ gadgets with $y\ge \mainRowEnd{}$ we can use is $j_2\block{}+\tau_1-t$.
Therefore we need to use at least $\tau_2\block{}+1-\mainRow{}-(j_2-j_1)\block{}-\tau_1+2t = 2t+3-(j_2-j_1)\block{}$ many $(y,x)$ gadgets with $y\in [\mainRow, \mainRowEnd)$.
If $j_2 < j_1$ then this is more than $\block{}$ gadgets.
But by construction, we would then need to use some $(y, j\block{}+2\tau_1+2)$ gadget, $y\in [\mainRow{}, \mainRowEnd{})$, which would incur an extra $\constEnabler{}$ cost. 
This implies that the shortest path would have a cost of at least $(n-m)\cdot U + (\tau_2\block{}+1)\constGadget{} + \constEnabler{}$, a contradiction.

\paragraph{$\constStayBigDiagonal{}$:}
The cost for using the $(j_1\block{}+t,0)$ gadget is at least $\constGadget{} + (\tau_2-j_1)\constStayBigDiagonal{}$, while the cost for using the $(\mainRowEnd{}+j_2\block{}+\tau_1-t-1, \tau_2\block{})$ gadget is $\constGadget{} + j_2 \cdot \constStayBigDiagonal{}$.

Due to Lemma~\ref{lem:upperBoundCost} and Assumption~\ref{as:gridStructure}, the cost of a shortest path must be smaller than $(n-m)\cdot U + (\tau_2\block{}+1)\constGadget{} + (\tau_2+1) \constStayBigDiagonal$.
As $j_2\ge j_1$, this implies that $j_2=j_1$.

\paragraph{$\constStaySmallDiagonal{}$, part I:}
We claim that a shortest path starting from the $(j_1\block{}+t+1,1)$ gadget uses all gadgets diagonally until it reaches the $(\mainRow{}, (\tau_2-j_1-1)\block{} + \tau_1 - t)$ gadget.
Suppose this is not the case, then there exists some first $(y, (\tau_2-j_1-1)\block{} + \tau_1 - t)$ gadget, $y>\mainRow{}$, reached by the shortest path.
\begin{itemize}
    \item If $y\le \mainRowEnd{}$, then the shortest path used the $(y-1, (\tau_2-j_1-1)\block{} + \tau_1 - t -1 )$ gadget, and the cost was at least $\constGadget{} + \constStaySmallDiagonal$.
In this case we could improve the shortest path, by first moving diagonally from the top left corner of the $(j_1\block{}+t+1,1)$ gadget to the top left corner of the $(\mainRow{}, (\tau_2-j_1-1)\block{} + \tau_1 - t)$ gadget, and then vertically to the $(y, (\tau_2-j_1-1)\block{} + \tau_1 - t)$ gadget. This would use the same number of gadgets (but all of them would have cost $\constGadget{}$, while in the original path at least one gadget costs at least an extra $\constStaySmallDiagonal{}$) and the same number of vertical edges.
    \item If $y> \mainRowEnd{}$, then the path cannot possibly reach the $(\mainRowEnd{}+j_1\block{}+\tau_1-t-1, \tau_2\block{})$ gadget.
This is because even if it takes only diagonal edges after reaching the $(y, (\tau_2-j_1-1)\block{} + \tau_1 - t)$ gadget, it reaches the $(y+j_1\block{}+\block{}-\tau_1+t, \tau_2\block{})$ gadget. But then it cannot reach the $(\mainRowEnd{}+j_1\block{}+\tau_1-t-1, \tau_2\block{})$ gadget, as $y+j_1\block{}+\block{}-\tau_1+t > \mainRowEnd{} +j_1\block{}+(2\tau_1+2)-\tau_1+t > \mainRowEnd{}+j_1\block{}+\tau_1-t-1$.
\end{itemize}

A completely symmetrical argument shows that the shortest path reaches the $(\mainRowEnd{}, (\tau_2-j_1-1)\block{}+\tau_1+t+3)$ gadget, moves diagonally to the top left corner of the $(\mainRowEnd{}+j_1\block{}+\tau_1-t-1, \tau_2\block{})$ gadget, uses it, and then moves vertically to $(n-1,m-1)$.

\paragraph{$\constStaySmallDiagonal{}$, part II}
So far we proved that for some $j_1\in [\tau_2]$ a shortest path reaches the $(\mainRow{}, (\tau_2-j_1-1)\block{}+\tau_1-t)$ gadget, and then moves to the top left corner of the $(\mainRowEnd{}, (\tau_2-j_1-1)\block{}+\tau_1+t+3)$ gadget.
Therefore it needs to use some $(y,(\tau_2-j_1-1)\block{}+\tau_1+1)$ gadget, $y\in [\mainRow{}, \mainRowEnd{})$.
Notice that if this gadget had cost at least $\constGadget{} + \constEnabler{}$, the shortest path would have cost at least $(n-m)\cdot U + (\tau_2\block{}+1)\constGadget{} + \constEnabler{}$; but this contradicts Lemma~\ref{lem:upperBoundCost}.
Therefore it must use some $(\mainRow{}+\tau_1+l_0, (\tau_2-j_1-1)\block{}+\tau_1+1)$ gadget, for some $l_0\in [\tau_3]$.

Additionally, the cheapest way to move from the top left corner of the $(\mainRow{}, (\tau_2-j_1-1)\block{}+\tau_1-t)$ gadget to the top left corner of the $(\mainRow{}+\tau_1+l_0, (\tau_2-j_1-1)\block{}+\tau_1+1)$ gadget is by moving vertically to the top left corner of the $(\mainRow{}+\tau_1-t-1+l_0, (\tau_2-j_1-1)\block{}+\tau_1-t)$ gadget and then diagonally to the $(\mainRow{}+\tau_1+l_0, (\tau_2-j_1-1)\block{}+\tau_1+1)$ gadget.
The reason is that for each $x\in [t+1]$, the shortest path must use some $(\mainRow{}+y, (\tau_2-j_1-1)\block{}+\tau_1-t+x)$ gadget, for $y\in [\mainHeight{}]$.
It cannot be $y>\tau_1-t-1+l_0+x$, because then the top left corner of the $(\mainRow{}+\tau_1+l_0, (\tau_2-j_1-1)\block{}+\tau_1+1)$ gadget would be unreachable.
If $y<\tau_1-t-1+l_0+x$, then the cost is at least $\constGadget{} + (\mainHeight{}-l_0+1)\constStaySmallDiagonal{}$, while the cost of using $y=\tau_1-t-1+l_0+x$ is less than $\constGadget{} + (\mainHeight{}-l_0+1)\constStaySmallDiagonal{}$.
Therefore the suggested path uses the gadget with the smallest cost, for every $x$.

With a completely symmetrical argument, when the shortest path reaches the $(\mainRow{}+\tau_1+l_0, (\tau_2-j_1-1)\block{}+\tau_1+1)$ gadget, it continues diagonally to the $(\mainRow{}+\tau_1+l_0+t+2, (\tau_2-j_1-1)\block{}+\tau_1+t+3)$ gadget, and then vertically to the $(\mainRowEnd{}, (\tau_2-j_1-1)\block{}+\tau_1+t+3)$ gadget.

Putting it all together, we conclude that the shortest path is in fact the $(j_1,l_0)$-restricted path.
\end{proof}

\subsection{Gadget implementation} \label{ssec:implementGadgets}
For a gadget $(\eta, \theta)$, we say that its $i$-th row is row $\eta\gadgetSize{}+i$, and similarly its $j$-th column is column $\theta\gadgetSize{}+j$.

Let \% denote the modulo operator.
We use $r_i = i\% \gadgetSize{}, c_j=j\% \gadgetSize{}$, for all $i,j$.
Therefore the $i$-th row of a gadget always has the same identifier, and similarly for the $j$-th column.
We say that the corresponding column of the $i$-th row is the $i$-th column, and vice versa.

We now give a high level overview of the weights of all rows and the activations of all columns.
Suppose we have an $(\eta, \theta)$ gadget.
\begin{itemize}
	\item The first row has weight $(n-\eta\gadgetSize{})\belowConstGadget{}$.
Along with the last row, these two ensure that the gadget is an $(n-\eta\gadgetSize{})\belowConstGadget{}$ gadget.
	\item The next $\consts{} = 8$ rows of every gadget have weights that are independent of each other.
Each of them describe a different constant.
	\item The next $\binaryNum{} = \ceil{\log_2{\constStayBigDiagonal{}}}$ rows of every gadget can be thought of as one block.
The weight of each row is a different power of $2$.
The intuition is that we can encode numbers by activating the proper columns.
	\item The next $\nodes{}=N+2$ rows of every gadget can be thought of as one block.
Each of the topmost $\mainRow{}$ and the bottommost $\mainRow{}$ gadgets is associated with some set of nodes, and this block encodes this set.
	\item The next $\smallCliqueGadgetA{}=(|\mathcal{V}_1|+|\mathcal{V}_2|)|\mathcal{V}_4|(N+2)$ rows can be thought of as one block.
It describes the cost between two node sets.
One of them relates to the rows intersecting the gadget and contains one node in every part of $\mathcal{V}_1\cup \mathcal{V}_2$.
The other node set relates to the columns intersecting the gadget and contains one node in every part of $\mathcal{V}_4$.
This block is further subdivided into sub-blocks.
The cost of each sub-block is equal to the weight of an edge across the two node sets.
    \item The next $\smallCliqueGadgetB{}=|\mathcal{V}_3||\mathcal{V}_4|(N+2)$ rows can be thought of as one block, similar to the previous one.
This time the rows relate to a node set with one node in every part of $\mathcal{V}_3$.
	\item The next $\cliqueGadget{}=\overline{N}(|\mathcal{V}_1|+|\mathcal{V}_2|)|\mathcal{V}_3|(N+2)$ rows can be thought of as $\overline{N}$ blocks.
Each block describes the cost between two node sets.
One of them relates to the diagonals (not the rows) intersecting the gadget and contains one node in every part of $\mathcal{V}_3$.
The other node set relates to the columns intersecting the gadget and contains one node in every part of $\mathcal{V}_1\cup \mathcal{V}_2$.
Each block is further subdivided into sub-blocks, describing the weight of an edge across the two node sets.
From a high-level view, the idea is that using $\overline{N}$ blocks we simulate different weights per row, which are essential in order to relate one of the node sets to the diagonals, instead of the rows.
    \item The next $\cliqueGadget{}$ rows is a similar block, for technical reasons.
	\item The last row has weight $\constGadget{} - (n-\eta\gadgetSize{})\belowConstGadget{}$.
\end{itemize}

Therefore, as claimed, $\gadgetSize{} = 2+\consts{} + \binaryNum{} + \nodes{} + \smallCliqueGadgetA{} + \smallCliqueGadgetB{} + 2\cliqueGadget{} = 12+ \ceil{\log_2{\constStayBigDiagonal{}}} + N + (|\mathcal{V}_1|+|\mathcal{V}_2|+|\mathcal{V}_3|)|\mathcal{V}_4|(N+2) + 2\overline{N}(|\mathcal{V}_1|+|\mathcal{V}_2|)|\mathcal{V}_3|(N+2)$.

We now define the weights of the rows more formally.
It is straightforward to verify that for every row $i$ its weight $d_i$ is non-negative, less than $2\constGadget{}$, and in case neither $i$ nor $i+1$ are a multiple of $g$ then $d_i \le n \constEnabler{}$, as required by Assumptions~\ref{as:gridStructure} and \ref{as:gadgets}.

We note that the weights of the rows at phase $(s,t)$ do not depend on $s$ or $t$.
When we do not specify the weight of a row, it is implied that its weight is $0$.

Let $(\eta, \theta)$ be a gadget.
\paragraph{First row of $(\eta,\theta)$ gadget:}
The weight of this row is  $(n-\eta\gadgetSize{})\belowConstGadget{}$.

\paragraph{The next $\consts{}$ rows of $(\eta,\theta)$ gadget:}
\begin{itemize}
	\item If $\eta=j\block{}+i$, for some $j\in [\tau_2], i\in [\tau_1]$, then the first row has weight $(\tau_2-j)\constStayBigDiagonal{} + (g-(|\mathcal{V}_1|+|\mathcal{V}_2|)|\mathcal{V}_4|)2M + w(U_1(\tau_1-i-1)\circ U_2(\tau_2-j-1), U_1(\tau_1-i-1)\circ U_2(\tau_2-j-1))$.
    Otherwise the first row has weight $\constEnabler{}$.
	\item If $\eta=\mainRowEnd{}+j\block{}+\tau_1-i-1$, for some $j\in [\tau_2], i\in [\tau_1]$, then the second row has weight $j\cdot \constStayBigDiagonal{}$.
    Otherwise the second row has weight $\constEnabler{}$.
	\item If $\eta = \mainRow{}+y$, for some $y\in [\mainHeight{}]$, then the third row has weight $(\mainHeight{}-y)\constStaySmallDiagonal{} + (\gadgetSize{}-(|\mathcal{V}_1|+|\mathcal{V}_2|)|\mathcal{V}_3|)2M$.
    \item If $\eta = \mainRow{}+y$, for some $y\in [\mainHeight{}]$, then the fourth row has weight $(\gadgetSize{}-(|\mathcal{V}_1|+|\mathcal{V}_2|)|\mathcal{V}_3|)2M$.
    \item If $\eta = \mainRow{}+y$, for some $y\in [\mainHeight{}]$, then the fifth row has weight $(\mainHeight{}+y-2\tau_1)\constStaySmallDiagonal{} + (\gadgetSize{}-(|\mathcal{V}_1|+|\mathcal{V}_2|)|\mathcal{V}_3|)2M$.
    \item If $\eta = \mainRow{}+y$, for some $y\in [\tau_1]\cup [\tau_1+\tau_3,\mainHeight{})$, then the sixth row has weight $\constEnabler{}$.
    \item If $\eta = \mainRow{}+\tau_1+l$, for $l\in [\tau_3]$, then the seventh row has weight $(\gadgetSize{}-|\mathcal{V}_3||\mathcal{V}_4|)2M + w(U_3(l), U_3(l))$.

    \item If $\eta = \mainRow{}+y$, for some $y\in [\mainHeight{}]$, then the eighth row has weight $\constEnabler$.
\end{itemize}

\paragraph{The next $\binaryNum{}$ rows of $(\eta,\theta)$ gadget:}
If $\eta \in [\mainRow{}, \mainRowEnd{})$, then the $i$-th of these rows has weight $2^{i-1}$.

\paragraph{The next $\nodes{}$ rows of $(\eta,\theta)$ gadget:}
If $\eta=j\block{} + i$ or $\eta=\mainRowEnd{}+j\block{}+\tau_1-i-1$ for some $i\in [\tau_1], j\in [\tau_2]$, then the $u$-th row has weight $\constEnabler{}$ for every $u\in U_1(\tau_1-i-1)$.
Else the $\alpha$-th row has weight $\constEnabler{}$.

\paragraph{The next $\smallCliqueGadgetA{}$ rows of $(\eta,\theta)$ gadget:}
If $\eta=j\block{} + i$ for some $i\in [\tau_1], j\in [\tau_2]$, then let $u_0, u_1, \ldots, u_{|\mathcal{V}_1|-1}$ be the sequence $U_1(\tau_1-i-1)$ and $u_{|\mathcal{V}_1|}, u_{|\mathcal{V}_1|+1}, \ldots, u_{|\mathcal{V}_1|+|\mathcal{V}_2|-1}$ be the sequence $U_2(\tau_2-j-1)$.
The weight of the  $i'|\mathcal{V}_4|(N+2) + j'(N+2) + u$-th row, $i'\in [|\mathcal{V}_1|+|\mathcal{V}_2|], j'\in [|\mathcal{V}_4|], u\in [N+2]$, is equal to $2M+w(u_{i'},u)$.

\paragraph{The next $\smallCliqueGadgetB{}$ rows of $(\eta,\theta)$ gadget:}
If $\eta=\mainRow{}+\tau_1+l$ for some $l\in [\tau_3]$, then let $u_0, u_1, \ldots, u_{|\mathcal{V}_3|-1}$ be the sequence $U_3(l)$.
The weight of the  $i'|\mathcal{V}_4|(N+2) + j'(N+2) + u$-th row, $i'\in [|\mathcal{V}_3|], j'\in [|\mathcal{V}_4|], u\in [N+2]$, is equal to $2M+w(u_{i'},u)$.

\paragraph{The next $\cliqueGadget{}$ rows of $(\eta,\theta)$ gadget:}
These rows are non-zero only if $\eta=\mainRow{}+y$ for some $y\in [\mainHeight{}]$.

For $i' \in [|\mathcal{V}_3|], i\in [\overline{N}]$, let $z_{i'}\in [\overline{N}]$ be the number described between bits $\log_2{\overline{N}}\cdot i'$ and $\log_2{\overline{N}}\cdot (i'+1) - 1$ of $y$.
If $i'<|\mathcal{V}_1|$, $(z_{i'}-i) \in V_{|\mathcal{V}_1|+|\mathcal{V}_2|+i'}$ and $i \in V_{i'}$, then let $u_{i',i} = z_{i'} - i$.
Else if $i'\ge|\mathcal{V}_1|$ and $z_{i'} \in V_{|\mathcal{V}_1|+|\mathcal{V}_2|+i'}$ let $u_{i',i} = z_{i'}$.
Else let $u_{i',i} = \beta$.

The weight of the $i|\mathcal{V}_3|(|\mathcal{V}_1|+|\mathcal{V}_2|)(N+2) + i'(|\mathcal{V}_1|+|\mathcal{V}_2|)(N+2) + j'(N+2) + u$-th row, $i\in [\overline{N}], i'\in [|\mathcal{V}_3|], j'\in [|\mathcal{V}_1|+|\mathcal{V}_2|], u\in [N+2]$, is equal to $2M+w(u_{i',i},u)$.

\paragraph{The next $\cliqueGadget{}$ rows of $(\eta,\theta)$ gadget:}
They have the same weights with the previous $\cliqueGadget{}$ rows, with the only difference being that the sign of the $w(\cdot,\cdot)$ term is flipped.

More formally, these rows are non-zero only if $\eta=\mainRow{}+y$ for some $y\in [\mainHeight{}]$.

For $i' \in [|\mathcal{V}_3|], i\in [\overline{N}]$, let $z_{i'}\in [\overline{N}]$ be the number described between bits $\log_2{\overline{N}}\cdot i'$ and $\log_2{\overline{N}}\cdot (i'+1) - 1$ of $y$.
If $i'<|\mathcal{V}_1|$, $(z_{i'}-i) \in V_{|\mathcal{V}_1|+|\mathcal{V}_2|+i'}$ and $i \in V_{i'}$, then let $u_{i',i} = z_{i'} - i$.
Else if $i'\ge|\mathcal{V}_1|$ and $z_{i'} \in V_{|\mathcal{V}_1|+|\mathcal{V}_2|+i'}$ let $u_{i',i} = z_{i'}$.
Else let $u_{i',i} = \beta$.

The weight of the $i|\mathcal{V}_3|(|\mathcal{V}_1|+|\mathcal{V}_2|)(N+2) + i'(|\mathcal{V}_1|+|\mathcal{V}_2|)(N+2) + j'(N+2) + u$-th row, $i\in [\overline{N}], i'\in [|\mathcal{V}_3|], j'\in [|\mathcal{V}_1|+|\mathcal{V}_2|], u\in [N+2]$, is equal to $2M-w(u_{i',i},u)$.

\paragraph{Last row of $(\eta,\theta)$ gadget:}
The weight of this row is $\constGadget{} - (n-\eta\gadgetSize{})\belowConstGadget{}$.

\paragraph{}
Similarly, we define the activations of the columns of a gadget $(\eta, \theta)$ at phase $(s,t)$.
In this case, the activations of columns may depend on $s$ or $t$.
Whenever we do not specify some column, it is implied that it is not active.

\paragraph{First column of $(\eta,\theta)$ gadget:}
This column is always active.

\paragraph{The next $\consts{}$ columns of $(\eta,\theta)$ gadget:}
\begin{itemize}
    \item If $\theta=0$ then the first column is active.
    \item If $\theta=\tau_2\block{}$ then the second column is active.
    \item If $\theta=j\block{}+i+1$ for some $i\in [\tau_1], j\in [\tau_2]$, then the third column is active.
    \item If $\theta=j\block{}+i+1$ for some $i\in [\tau_1-1], j\in [\tau_2]$, then the fourth column is active
    \item If $\theta=j\block{}+\tau_1+i+2$ for some $i\in [\tau_1], j\in [\tau_2]$, then the fifth column is active.
    \item If $\theta=j\block{}+\tau_1+1$ for some $j\in [\tau_2]$, then the sixth and seventh columns are active.
    \item If $\theta=j\block{}+2\tau_1+2$ for some $j\in [\tau_2]$, then the eighth column is active.
\end{itemize}

\paragraph{The next $\binaryNum{}$ columns of $(\eta,\theta)$ gadget:}
\begin{itemize}
    \item If $\theta=j\block{}+i+1$ for some $i\in [\tau_1], j\in [\tau_2]$, then the activation of the $x$-th column is equal to the $x$-th bit in the binary representation of $i\cdot \constStaySmallDiagonal{}$.
    \item If $\theta=j\block{}+\tau_1+i+2$ for some $i\in [\tau_1], j\in [\tau_2]$, then the activation of the $x$-th column is equal to the $x$-th bit in the binary representation of $(\tau_1-i-1)\constStaySmallDiagonal{}$.
\end{itemize}

\paragraph{The next $\nodes{}$ columns of $(\eta,\theta)$ gadget:}
If $\theta=0$ or $\theta=\tau_2\block{}$, then the $u$-th column is activated if and only if $u \not \in U_1(\tau_1-t-1)$.
We note that this is the only case where the activation of columns depends on $t$.

\paragraph{The next $\smallCliqueGadgetA{}$ columns of $(\eta,\theta)$ gadget:}
If $\theta=0$ then let $v_0, v_1, \ldots, v_{|\mathcal{V}_4|-1}$ be the sequence $U_4(s)$.
The $i'|\mathcal{V}_4|(N+2) + j'(N+2) + u_{j'}$-th column is activated, $i'\in [|\mathcal{V}_1|+|\mathcal{V}_2|], j'\in [|\mathcal{V}_4|]$.

\paragraph{The next $\smallCliqueGadgetB{}$ columns of $(\eta,\theta)$ gadget:}
If $\theta=j\block{}+\tau_1+1$, $j\in [\tau_2]$, then let $v_0, v_1, \ldots, v_{|\mathcal{V}_4|-1}$ be the sequence $U_4(s)$.
The $i'|\mathcal{V}_4|(N+2) + j'(N+2) + u_{j'}$-th column is activated, $i'\in [|\mathcal{V}_3|], j'\in [|\mathcal{V}_4|]$.

\paragraph{The next $\cliqueGadget{}$ columns of $(\eta,\theta)$ gadget:}
These columns are non-zero only if $\theta=j\block{}+i+1$ for some $i\in [\tau_1], j\in [\tau_2]$.

For $j'\in [|\mathcal{V}_1|]$, let $x_{i,j'}$ be the number described between bits $\log_2{\overline{N}}\cdot j'$ and $\log_2{\overline{N}}\cdot (j'+1) - 1$ of $i$.
Let $u_{i,j'}=x_{i,j'}$ if $x_{i,j'}\in V_{i'}$, and $u_{i,j'} = \alpha$ otherwise.
We activate the columns $x_{i,j'}|\mathcal{V}_3|(|\mathcal{V}_1|+|\mathcal{V}_2|)(N+2) + i'(|\mathcal{V}_1|+|\mathcal{V}_2|)(N+2) + j'(N+2) + u_{i,j'}$, $i'\in [|\mathcal{V}_3|]$.

For $j'\in [|\mathcal{V}_1|, |\mathcal{V}_1|+|\mathcal{V}_2|)$, let $v_{j'}$ be the $(j'-|\mathcal{V}_1|)$-th node in $U_2(j)$.
We activate the columns $i'(|\mathcal{V}_1|+|\mathcal{V}_2|)(N+2) + j'(N+2) + u_{j'}$, $i'\in [|\mathcal{V}_3|]$.

\paragraph{The next $\cliqueGadget{}$ columns of $(\eta,\theta)$ gadget:}
These columns are non-zero only if $\theta=j\block{}+i+1$ for some $i\in [\tau_1-1], j\in [\tau_2]$.
They have the same activations with the previous $\cliqueGadget{}$ columns, with the only difference being that we use $(i+1)$ wherever we previously used $i$, and that $i\in [\tau_1-1]$ instead of $i \in [\tau_1]$.

For $j'\in [|\mathcal{V}_1|]$, let $x_{(i+1),j'}$ be the number described between bits $\log_2{\overline{N}}\cdot j'$ and $\log_2{\overline{N}}\cdot (j'+1) - 1$ of $(i+1)$.
Let $u_{(i+1),j'}=x_{(i+1),j'}$ if $x_{(i+1),j'}\in V_{i'}$, and $u_{(i+1),j'} = \alpha$ otherwise.
We activate the columns $x_{(i+1),j'}|\mathcal{V}_3|(|\mathcal{V}_1|+|\mathcal{V}_2|)(N+2) + i'(|\mathcal{V}_1|+|\mathcal{V}_2|)(N+2) + j'(N+2) + u_{(i+1),j'}$, $i'\in [|\mathcal{V}_3|]$.

For $j'\in [|\mathcal{V}_1|, |\mathcal{V}_1|+|\mathcal{V}_2|)$, let $v_{j'}$ be the $(j'-|\mathcal{V}_1|)$-th node in $U_2(j)$.
We activate the columns $i'(|\mathcal{V}_1|+|\mathcal{V}_2|)(N+2) + j'(N+2) + u_{j'}$, $i'\in [|\mathcal{V}_3|]$.

\paragraph{Last column of $(\eta,\theta)$ gadget:}
This column is always activated.

\paragraph{Gadgets' costs}
We are now ready to prove that with the defined row weights and column activations, the costs of our gadgets are the desired ones.

\begin{lemma} \label{lem:implementCosts}
At phase $(s,t)$, where $U_1(\tau_1-t-1)\not\ni \alpha$, the cost of any gadget is the desired cost, according to Definition~\ref{def:desiredCosts}.
\end{lemma}
\begin{proof}
For every gadget, the sum of the weight of its first and its last row is $\constGadget{}$.
Furthermore, the first and the last column of every gadget is activated.
Therefore, the first and the last row contribute a cost of $\constGadget{}$.

We begin with the cases where the desired cost is related to the weights of edges in $G_0$, as these are the most technically challenging ones.

\paragraph{The $(j\block{}+t,0)$ gadgets, $j\in [\tau_2]$.}
~\\Desired cost: $\constGadget{} + (\tau_2-j)\constStayBigDiagonal{} + 2\gadgetSize{}M +
w(U_1(\tau_1-t-1)\circ U_2(\tau_2-j-1), U_1(\tau_1-t-1)\circ U_2(\tau_2-j-1)\circ U_4(s))$.

From rows $1, \ldots, \consts{}$, only the first row has both non-zero weight $(\tau_2-j)\constStayBigDiagonal{} + (g-(|\mathcal{V}_1|+|\mathcal{V}_2|)|\mathcal{V}_4|)2M + w(U_1(\tau_1-t-1)\circ U_2(\tau_2-j-1), U_1(\tau_1-t-1)\circ U_2(\tau_2-j-1))$, and the corresponding (first) column is activated.

The next $\binaryNum{}$ rows all have weight $0$.

For the next $\nodes{}$ rows, we have non-zero weight $(\constEnabler{})$ in row $u$ if and only if $u\in U_1(\tau_1-t-1)$.
However, in these cases the corresponding columns are deactivated, therefore the total contribution is zero.

Out of the next $\smallCliqueGadgetA{}$ rows, the corresponding columns activated are the $i'|\mathcal{V}_4|(N+2) + j'(N+2) + u_{j'}$, where $i'\in [|\mathcal{V}_1|+|\mathcal{V}_2|], j'\in [|\mathcal{V}_4|]$, and $u_{j'}$ is the $j'$-th node of $U_4(s)$.
For each such column, the weight of the corresponding row is equal to the weight of the edge between the $i'$-th node of $U_1(\tau_1-t-1)\circ U_2(\tau_2-j-1)$ and $u_{j'}$, plus $2M$.
Summing up these costs gives $(|\mathcal{V}_1|+|\mathcal{V}_2|)|\mathcal{V}_4|2M + w(U_1(\tau_1-t-1)\circ U_2(\tau_2-j-1), U_4(s))$.

The next $\smallCliqueGadgetB{}+2\cliqueGadget{}$ rows all have zero weight.

Summing up all the costs, we get that the cost of such a gadget is $\constGadget{} + (\tau_2-j)\constStayBigDiagonal{} + 2\gadgetSize{}M + w(U_1(\tau_1-t-1)\circ U_2(\tau_2-j-1), U_1(\tau_1-t-1)\circ U_2(\tau_2-j-1)\circ U_4(s))$.

\paragraph{The $(\mainRow{}+i+l, j\block+i+1)$ gadgets, with $i\in [\tau_1], j\in [\tau_2], l\in [\tau_3]$.}
~\\Desired cost if $i<\tau_1-1$: $\constGadget{} + (\mainHeight{}-l)\constStaySmallDiagonal + 4\gadgetSize{}M + h(i,j,l) - h(i+1,j,l)$. 
~\\Desired cost if $i=\tau_1-1$: $\constGadget{} + (\mainHeight{}-l)\constStaySmallDiagonal + 2\gadgetSize{}M + h(\tau_1-1,j,l)$.

From rows $1, \ldots, \consts{}$, if $i=\tau_1-1$ only the third row has both non-zero weight ($(\mainHeight{}-i-l)\constStaySmallDiagonal{} + (g-(|\mathcal{V}_1|+|\mathcal{V}_2|)|\mathcal{V}_3|)2M$) and the corresponding (third) column is activated. Else only the third and the fourth row have both non-zero weight ($(\mainHeight{}-i-l)\constStaySmallDiagonal{} + (g-(|\mathcal{V}_1|+|\mathcal{V}_2|)|\mathcal{V}_3|)2M$ and $(g-(|\mathcal{V}_1|+|\mathcal{V}_2|)|\mathcal{V}_3|)2M$) and the corresponding columns (third and fourth) are activated.

Out of the next $\binaryNum{}$ rows, the $x$-th of them has weight $2^{x-1}$ and the  corresponding column is active if and only if the $x$-th bit in the binary representation of $i\cdot \constStaySmallDiagonal{}$ is $1$.
Therefore the total cost from these rows is $i\cdot \constStaySmallDiagonal{}$.

The next $\nodes{}+\smallCliqueGadgetA{}$ rows all have zero weight.

The next $\smallCliqueGadgetB{}$ rows all have corresponding columns that are not activated.

The analysis for the rest of the rows is the most technically challenging part of the proof.
Out of the next $\cliqueGadget{}$ rows, we examine all the corresponding columns that are activated.
We take two cases:
\begin{itemize}
    \item For $j'\in [|\mathcal{V}_1|]$, let $x_{i,j'}$ be the number described between bits $\log_2{\overline{N}}\cdot j'$ and $\log_2{\overline{N}}\cdot (j'+1) - 1$ of $i$.
    Let $u_{i,j'}=x_{i,j'}$ if $x_{i,j'}\in V_{i'}$, and $u_{i,j'} = \alpha$ otherwise.
    In other words, $u_{i,j'}$ is the $j'$-th node of $U_1(i)$.
    
    For $i'\in [|\mathcal{V}_3|]$, columns $x_{i,j'}|\mathcal{V}_3|(|\mathcal{V}_1|+|\mathcal{V}_2|)(N+2) + i'(|\mathcal{V}_1|+|\mathcal{V}_2|)(N+2) + j'(N+2) + u_{i,j'}$ are activated.
    
    Let us fix $i'\in [|\mathcal{V}_3|]$.
    To describe the weight of the corresponding row, first let $z_{i',i+l}\in [\overline{N}]$ be the number described between bits $\log_2{\overline{N}}\cdot i'$ and $\log_2{\overline{N}}\cdot (i'+1) - 1$ of $i+l$.
    If $i'<|\mathcal{V}_1|$, $(z_{i',i+l}-x_{i,j'}) \in V_{|\mathcal{V}_1|+|\mathcal{V}_2|+i'}$ and $x_{i,j'} \in V_{i'}$, then let $v_{i',x_{i,j'},i+l} = z_{i',i+l} - x_{i,j'}$.
    Else if $i'\ge|\mathcal{V}_1|$ and $z_{i',i+l} \in V_{|\mathcal{V}_1|+|\mathcal{V}_2|+i'}$ let $v_{i',x_{i,j'},i+l} = z_{i',i+l}$.
    Else let $v_{i',x_{i,j'},i+l} = \beta$.
    
    The weight of row $x_{i,j'}|\mathcal{V}_3|(|\mathcal{V}_1|+|\mathcal{V}_2|)(N+2) + i'(|\mathcal{V}_1|+|\mathcal{V}_2|)(N+2) + j'(N+2) + u_{i,j'}$ is equal to $2M+w(v_{i',x_{i,j'},i+l}, u_{i,j'})$.

    Notice that $\sum_{p\in [|\mathcal{V}_1|], q\in [|\mathcal{V}_3|]} w(v_{q,x_{i,p},i+l}, u_{i,p})$ is a function of $i$ and $i+l$, therefore a function of $i$ and $l$, which we call $g(i,l)$.
    It holds that $g(i,l) \le |\mathcal{V}_1||\mathcal{V}_3|2M$ for any $i,l$.
    
    Furthermore, assume that $U_1(i)\not \ni \alpha$ and $U_3(l) \not \ni \beta$ (thus $u_{i,j'} = x_{i,j'}$). 
    Then the $(p+1)\log_2{\overline{N}}-1$ bit of both $U_1(i)$ and $U_3(l)$ is always $0$ for any $p$, meaning that when we add $i$ and $l$, there is no carry from the $(p+1)\log_2{\overline{N}}-1$ to the $(p+1)\log_2{\overline{N}}$ bit.
    In effect, the number between bits $i'\log_2{\overline{N}}$ and $(i'+1)\log_2{\overline{N}}-1$ of $i+l$ is the same as the sum of the number between bits $i'\log_2{\overline{N}}$ and $(i'+1)\log_2{\overline{N}}-1$ of $i$ and the number between bits $i'\log_2{\overline{N}}$ and $(i'+1)\log_2{\overline{N}}-1$ of $l$.    
    Viewing it the other way around, the number between bits $i'\log_2{\overline{N}}$ and $(i'+1)\log_2{\overline{N}}-1$ of $l$ (the $i'$-th node in $U_3(l)$) is equal to the number between bits $i'\log_2{\overline{N}}$ and $(i'+1)\log_2{\overline{N}}-1$ of $i+l$ minus the number between bits $i'\log_2{\overline{N}}$ and $(i'+1)\log_2{\overline{N}}-1$ of $i$ (this difference is exactly $v_{i',x_{i,j'},i+l}$).

    We conclude that if $U_1(i)\not \ni \alpha$ and $U_3(l) \not \ni \beta$, then $w(v_{i',x_{i,j'},i+l}, u_{i,j'})$ is the weight of the edge between the $j'$-th node of $U_1(i)$ and the $i'$-th node of $U_3(l)$.
    Therefore $g(i,l) = w(U_1(i), U_3(l)) \le M$.

    We now show that if $\alpha \in U_1(i)$ or $\beta \in U_3(l)$, then $g(i,l)$ is too large.
    This is later used to ensure the properties of $h(\cdot, \cdot, \cdot)$ specified by Definition~\ref{def:desiredCosts}.
    \begin{itemize}
        \item If $\alpha \in U_1(i)$ then $g(i,l)$ contains at least $|\mathcal{V}_3|$ terms that are $2M$, and the sum of all negative terms is at least $-M$, by definition of $M$.
        \item Similarly, if $v_{q,x_{i,p},i+l} = \beta$, for any $p,q$, then $g(i,l)$ has $|\mathcal{V}_1|$ terms that are $2M$, and the sum of all negative terms is at least $-M$, by definition of $M$.
        \item If $\alpha\not \in U_1(i)$ but $\beta \in U_3(l)$, let $r$ be the smallest term in $U_3(l)$ that is equal to $\beta$. Then, as we argued previously and by definition of $r$, it should be that for $r'\le r$ we have that the $r'$-th node in $U_3(l)$ is equal to $v_{r',x_{i,j'},i+l}$.
        Therefore $v_{r,x_{i,j'},i+l}= \beta$, and as in the previous case $g(i,l)$ has $|\mathcal{V}_1|$ terms that are $2M$, and the sum of all negative terms is at least $-M$, by definition of $M$.
    \end{itemize}

    \item For $j'\in [|\mathcal{V}_1|, |\mathcal{V}_1|+|\mathcal{V}_2|)$ the activated columns are the $i'(|\mathcal{V}_1|+|\mathcal{V}_2|)(N+2) + j'(N+2) + u_{j'}$, $i'\in [|\mathcal{V}_3|]$, where $u_{j'}$ is the $(j'-|\mathcal{V}_1|)$th node of $U_2(j)$.
    Notice that these columns are distinct from the ones activated in the previous case, as $u_{j'} \in V_{j'}$, while $u_{i,j}$ was always either $\alpha$ or in some $V_{j''}$ with $j''\in [|\mathcal{V}_1|]$.
    The weight of these rows is as previously, but now we have that $v_{i',x_{i,j'},i+l} = z_{i',i+l}$ if 
    $z_{i',i+l} \in V_{|\mathcal{V}_1|+|\mathcal{V}_2|+i'}$, and $v_{i',x_{i,j'},i+l} = \beta$ otherwise.

    The weight of row $i'(|\mathcal{V}_1|+|\mathcal{V}_2|)(N+2) + j'(N+2) + u_{j'}$ is equal to $2M+w(v_{i',x_{i,j'},i+l}, u_{j'})$.

    Notice that $\sum_{p\in [|\mathcal{V}_2], q\in [|\mathcal{V}_3|]} w(v_{q,x_{i,p+|\mathcal{V}_1|},i+l}, u_{p+|\mathcal{V}_1|})$ is a function of $i,j,l$ which we call $g'(i,j,l)$ (we now have a dependence on $j$ because of the definition of $u_{p+|\mathcal{V}_1|}$).
    This is upper bounded by $|\mathcal{V}_2||\mathcal{V}_3|2M$.
    
With the same arguments as previously, if $U_1(i)\not \ni \alpha$ and $U_3(l) \not \ni \beta$, then $g'(i,j,l) = w(U_2(j), U_3(l))$.  
\end{itemize}

Let $h(i,j,l) = g(i,l) + g'(i,j,l)$.
We conclude that the total cost is $(|\mathcal{V}_1|+|\mathcal{V}_2|)|\mathcal{V}_3|2M + h(i,j,l)$.
If $U_1(i)\not \ni \alpha$ and $U_3(l) \not \ni \beta$ then $h(i,j,l) = w(U_1(i)\circ U_2(j), U_3(l)) \le M$.
On the other hand, if $U_1(i)\ni \alpha$ or 
$U_3(l)\ni \beta$ then $h(i,j,l)\ge g(i,l) \ge \min\set{|\mathcal{V}_1|,|\mathcal{V}_3|}2M-M$, which is at least $2M$ for sufficiently large $k$.
Therefore, for any fixed $i,j$ such that $U_1(i)\not \ni \alpha$, it holds that there exists an $l_{i,j}$ such that $U_3(l_{i,j})\not \ni \beta$ and $h(i,j,l_{i,j}) \le h(i,j,l)$ for all $l$.
Finally, for any $i,j,l$ we upper bound $h(i,j,l)$ by $2N^2M$, using the upper bounds of $g$ and $g'$.

This proves the desired cost when $i=\tau_1-1$, as the next $\cliqueGadget{}$ rows all have non-activated corresponding columns.

When $i<\tau_1-1$, then we have the exact same analysis for the next $\cliqueGadget{}$ rows, with the only difference being that we use $i+1$ instead of $i$, and we reverse the sign of the weights.
Therefore we get an additional cost $(|\mathcal{V}_1|+|\mathcal{V}_2|)|\mathcal{V}_3|2M - h(i+1,j,l)$ from these rows.
Along with the additional $(\gadgetSize{} - (|\mathcal{V}_1|+|\mathcal{V}_2|)|\mathcal{V}_3|)2M$ cost from the fourth row of the gadget, this proves we indeed get the desired cost.

\paragraph{The $(\mainRow{}+\tau_1+l, j\block{}+\tau_1+1)$ gadgets, with $j\in [\tau_2], l\in [\tau_3]$.}
~\\Desired cost: $\constGadget{} + 2\gadgetSize{}M + w(U_3(l),U_3(l)\circ U_4(s))$.

From rows $1, \ldots, \consts{}$, only the seventh row has both non-zero weight $((\gadgetSize{}-|\mathcal{V}_3||\mathcal{V}_4|)2M + w(U_3(l), U_3(l)))$, and the corresponding (seventh) column is activated.

Out of the next $\binaryNum{}+\nodes{}+\smallCliqueGadgetA{}$ rows, all their corresponding columns are not activated.

Out of the next $\smallCliqueGadgetB{}$ rows, the corresponding columns activated are the $i'|\mathcal{V}_4|(N+2) + j'(N+2) + u_{j'}$, with $i'\in [|\mathcal{V}_3|], j'\in [|\mathcal{V}_4|]$, and $u_{j'}$ is the $j'$-th node of $U_4(s)$.
For each such column, the weight of the corresponding row is equal to the weight of the edge between the $i'$-th node of $U_3(l)$ and $u_{j'}$, plus $2M$.
Summing up these costs gives $|\mathcal{V}_3||\mathcal{V}_4|2M + w(U_3(l), U_4(s))$.

The next $2\cliqueGadget{}$ rows all have non-activated corresponding columns.

Summing up all the costs, we get $\constGadget{} + 2\gadgetSize{}M + w(U_3(l),U_3(l)\circ U_4(s))$.

\paragraph{The $(\mainRowEnd{} + j\block{}+\tau_1-t-1, \tau_2\block{})$ gadgets, $j\in [\tau_2]$.}
~\\Desired cost: $\constGadget{} + j\cdot \constStayBigDiagonal{}$.

From rows $1, \ldots, \consts{}$, only the second row has both non-zero weight $j\cdot \constStayBigDiagonal{}$ and its corresponding (second) column is activated.

The next $\binaryNum{}$ rows all have weight $0$.

Out of the next $\nodes{}$ rows, we have non-zero weight in row $u$ if and only if $u\in U_1(\tau_1-t-1)$.
However, in these cases the corresponding columns are deactivated, therefore the total contribution is zero.

The next $\smallCliqueGadgetA{} + \smallCliqueGadgetB+2\cliqueGadget{}$ rows all have weight $0$.

Therefore the cost of such a gadget is the desired cost $\constGadget{} + j\cdot \constStayBigDiagonal{}$.

\paragraph{The $(y_1,0)$ gadgets and the $(y_2, \tau_2\block{})$ gadgets, with $y_1<\mainRowEnd{}, y_1 \not \in \set{j\block{}+t \mid j\in [\tau_2]}, y_2\ge \mainRow{}, y_2\not \in \set{\mainRowEnd{} + j\block{}+\tau_1-t-1 \mid j\in [\tau_2]}$.}
~\\Desired cost: At least $\constGadget{} + \constEnabler{}$.

We only argue about the $(y_1, 0)$ gadgets, as the situation is similar for the $(y_2, \tau_2\block{})$ gadgets.
We only need a lower bound, therefore we can ignore rows $1, \ldots, \consts{} + \binaryNum{}$.

For the next $\nodes{}$ rows we take take two cases:
\begin{itemize}
    \item The $\alpha$-th row has weight $\constEnabler{}$.
    Notice that the corresponding column is activated, because we always assume $\alpha \not \in U_1(\tau_1-t-1)$.
    Therefore the gadget has cost at least $\constGadget{} + \constEnabler{}$.
    \item The $\alpha$-th row does not have weight $\constEnabler{}$.
    From the definition of row weights, this means $y_1=j\block+i$ for some $i\in [\tau_1], j\in [\tau_2]$.
    But as $i\ne t$, we have that $U_1(\tau_1-i-1) \ne U_1(\tau_1-t-1)$.
    Let $u$ be a node in $U_1(\tau_1-i-1)$ such that $u\not\in U_1(\tau_1-t-1)$.
    Then the $u$-th row has weight $\constEnabler{}$ and the corresponding column is activated, meaning that again the cost of the gadget is at least $\constGadget{} + \constEnabler{}$.
\end{itemize}

\paragraph{The $(y,x)$ gadgets, for $y<\mainRow{}$ or $y\ge \mainRowEnd{}$, and $0<x<\tau_2\block{}$.}
~\\Desired cost: $\constGadget{}$.

From rows $1, \ldots, \consts{}$, only the first and the second row have non-zero weight, but the first and second column are not activated.

The next $\binaryNum{}$ rows all have weight $0$.

Out of the next $\nodes{}+\smallCliqueGadgetA{}$ rows, all coresponding columns are not activated.

The next $\smallCliqueGadgetB{} + 2\cliqueGadget{}$ rows all have zero weight.

Therefore the cost of such a gadget is the desired cost $\constGadget{}$.

\paragraph{The $(\mainRow{}+y, j\block{}+i+1)$ gadgets, $i\in [\tau_1], j\in [\tau_2], y\in [\mainHeight{}] \setminus [i,i+\tau_3)$.}
~\\Desired cost: At least $\constGadget{} + (\mainHeight{}+i-y)\constStaySmallDiagonal$.

From rows $1, \ldots, \consts{}$, the third row has both non-zero weight (greater than $(\mainHeight{}-y)\constStaySmallDiagonal{}$) and the corresponding (third) column is activated. 

Out of the next $\binaryNum{}$ rows, the $x$-th of them has weight $2^{x-1}$ and the corresponding column is activated if and only if the $x$-th bit in the binary representation of $i\cdot \constStaySmallDiagonal{}$ is $1$.
Therefore the total cost from these rows is $i\cdot \constStaySmallDiagonal{}$.

This proves that the cost of such a gadget is at least $\constGadget{} + (\mainHeight{}+i-y)\constStaySmallDiagonal$.

\paragraph{The $(\mainRow{}+y, j\block{}+\tau_1+1)$ gadgets, with $y\in [\tau_1] \cup [\tau_1+\tau_3,\mainHeight{})$, and $j\in [\tau_2]$.}
~\\Desired cost: At least $\constGadget{} + \constEnabler{}$.

From rows $1, \ldots, \consts{}$, the sixth row has weight $\constEnabler{}$ and the sixth column is activated. 
Therefore the cost of such a gadget is at least $\constGadget{} + \constEnabler{}$.

\paragraph{The $(\mainRow{}+y, j\block{}+\tau_1+i+2)$ gadgets, with $i\in [\tau_1], j\in [\tau_2], y\in [\mainHeight{}]$.}
~\\Desired cost: $\constGadget{} + (\mainHeight{}-\tau_1-1+y-i) \constStaySmallDiagonal{}$.

From rows $1, \ldots, \consts{}$, only the fifth row has both non-zero weight ($(\mainHeight{}+y-2\tau_1)\constStaySmallDiagonal{}$) and the corresponding (fifth) column is activated. 

Out of the next $\binaryNum{}$ rows, the $x$-th of them has weight $2^{x-1}$ and the corresponding column is activated if and only if the $x$-th bit in the binary representation of $(\tau_1-i-1)\cdot \constStaySmallDiagonal{}$ is $1$.
Therefore the total cost from these rows is $(\tau_1-i-1)\cdot \constStaySmallDiagonal{}$.

Out of the next $\nodes{}+\smallCliqueGadgetA{}+\smallCliqueGadgetB{}+2\cliqueGadget{}$ rows, all their corresponding columns are deactivated.

We conclude that the cost of such a gadget is $\constGadget{} + (\mainHeight{}-\tau_1-1+y-i) \constStaySmallDiagonal{}$.

\paragraph{The $(\mainRow{}+y, j\block{}+2\tau_1+2)$ gadgets, with $j\in [\tau_2], y\in [\mainHeight{}]$.}
~\\Desired cost: At least $\constGadget{} + \constEnabler{}$.

From rows $1, \ldots, \consts{}$, the eighth row has weight $\constEnabler{}$ and the corresponding (eighth) column is activated. 

Therefore the cost of such a gadget is at least $\constGadget{} + \constEnabler{}$.
\end{proof}

\subsection{Lower bound}
We finally prove our lower bound for \Intermediary{}.
Recall that for $i\in [1,3]$ we have $|\mathcal{V}_i| = \floor{\rho_i k}, |\mathcal{V}_4| = k-|\mathcal{V}_1|-|\mathcal{V}_2|-|\mathcal{V}_3|$.
Additionally, for $i\in [1,4]$ we have $\tau_i = \Theta(N^{|\mathcal{V}_i|})$.
Finally, $n=n_{r}=\Theta(\tau_3+\tau_1\tau_2), m=n_c=\Theta(\tau_1\tau_2)$.

\lowerBoundIntermediary*
\begin{proof}
Let $\rho_1 = \frac{\beta\gamma}{(c+2)}, \rho_2 = \frac{(1-\beta)\gamma}{(c+2)}, \rho_3 = \frac{1}{(c+2)}$. Notice that with these $\rho_1, \rho_2, \rho_3$ values, and as $n=\Theta(\tau_3+\tau_1\tau_2)$, we get $n=\Theta(\tau_3)$.

For $i\in [|\mathcal{V}_1|]$ let $u_i = i \cdot N/k$, and let integer $p=f(u_0, \ldots, u_{|\mathcal{V}_1|-1})$ be the encoding of this sequence (therefore $U_1(p) = u_0, \ldots, u_{|\mathcal{V}_1|-1}$, and $U_1(p)\not \ni \alpha$).
Given a Negative-$k$-Clique instance, for a sufficiently large constant $k$, we use the reduction of \Cref{ssec:reduction} to formulate an instance of \Intermediary{}.
We properly set the activations of the columns so that we start at phase $(0,p)$.
For any given $s\in [\tau_4]$, we iterate over all phases $(s,t)$ with $t\in [\tau_1]$ and $U_1(\tau_1-t-1)\not \ni \alpha$ by properly updating our data structure.
In each phase we query our data structure.

Let $C_{s,t}$ be the minimum cost of a $k$-Clique in $G_0$ that includes all nodes in $U_1(\tau_1-t-1)$ and $U_4(s)$.
By Lemma~\ref{lem:structureSP} the shortest path at phase $(s,t)$ is a restricted path.
Therefore we can acquire the length of the shortest restricted path at phase $(s,t)$ by querying the data structure.
By Corollary~\ref{cor:costSP} we can retrieve $C_{s,t}$ given the length of the shortest restricted path.
As we iterate over all relevant $(s,t)$, we can compute the minimum cost of any $k$-Clique, which means we can decide whether there exists a Negative-$k$-Clique.

Concerning the running time, notice that we switch from a phase $(s,t)$ to a phase $(s,t')$ a total of $O(\tau_1 \tau_4)$ times, and we only need to update the columns of the $(y,0)$ and the $(y,\tau_2 \block{})$ gadgets, for any $y$.
There are at most $\gadgetSize{} = O(N^2)$ such columns.
We switch from a phase $(s,t)$ to a phase $(s',t')$ with $s'\ne s$ a total of $\tau_4-1$ times, and every time we need to update the columns of the $(y,0)$, the $(y,\tau_2\block{})$ and the $(y,j\block{}+\tau_1+1)$ gadgets, for $j\in [\tau_2]$ and all $y$.
There are $O(\tau_2 \gadgetSize{})$ such columns.

Therefore the time we spend to solve Negative-$k$-Clique is

\[O(T_p(n,m) + (\tau_1+\tau_2) \tau_4 N^2 T_u(n,m) + \tau_1\tau_4 T_q(n,m))\]

Assuming the Negative-$k$-Clique Hypothesis, for all $\delta'>0$ we have that

\begin{align*}
T_p(n,m) + (\tau_1+\tau_2) \tau_4 N^2 T_u(n,m) + \tau_1 \tau_4 T_q(n,m) &= \Omega(N^{k-\delta'}) \implies \\
T_p(n,m) + (\tau_1+\tau_2) \tau_4 N^2 T_u(n,m) + \tau_1 \tau_4 T_q(n,m) &= \Omega(\tau_1 \cdot \tau_2 \cdot \tau_3 \cdot \tau_4 N^{-\delta'}) \implies \\
T_p(n,m)/\tau_4 + (\tau_1+\tau_2) T_u(n,m) + \tau_1 T_q(n,m) &= \Omega(\tau_1 \cdot \tau_2 \cdot \tau_3 \cdot N^{-2-\delta'})
\end{align*}

We now have that:

\begin{align*}
T_p(n,m)/\tau_4 &= O((N^{\rho_3k}+N^{\rho_1k+\rho_2k})^c / N^{k-\floor{\rho_1 k}-\floor{\rho_3 k}-\floor{\rho_3 k}}) \\
&= O((N^{ck/(c+2)} / N^{(1-2/(c+2))k}) \\
&= O(1)
\end{align*}

Therefore the term $T_p(n,m)/\tau_4$ is negligible. As $\rho_1 \le \rho_2$, we get

\[\tau_2 T_u(n,m) + \tau_1 T_q(n,m) = \Omega(\tau_1 \cdot \tau_2 \cdot \tau_3 \cdot N^{-2-\delta'})\]

Thus either $\tau_2 T_u(n,m) = \Omega(\tau_1 \cdot \tau_2 \cdot \tau_3 \cdot N^{-2-\delta'})$ or $\tau_1 T_q(n,m) = \Omega(\tau_1 \cdot \tau_2 \cdot \tau_3 \cdot N^{-2-\delta'})$.

Assume $\tau_2 T_u(n,m) = \Omega(\tau_1 \cdot \tau_2 \cdot \tau_3 \cdot N^{-2-\delta'})$, then

\begin{align*}
    T_u(n,m) &= \Omega(\tau_1 \cdot \tau_3 \cdot N^{-2-\delta'}) \\ 
    &= \Omega(N^{\rho_1k-1} \cdot n \cdot N^{-2-\delta'})\\
    &= \Omega(\frac{m^{\beta}}{N} \cdot n \cdot N^{-2-\delta'})\\
    &= \Omega(m^{\beta} \cdot n \cdot N^{-3-\delta'})\\
\end{align*}

But $m = \Theta(N^{\floor{\rho_1 k}+\floor{\rho_2 k}})$, therefore there exists a sufficiently large $k$ such that $N^{-3-\delta'} = \Omega(m^{-\delta})$, which gives us that

\[T_u(n,m) = \Omega(n \cdot m^{\beta-\delta})\]

Repeating the same arguments gives that if $\tau_1 T_q(n,m) = \Omega(\tau_1 \cdot \tau_2 \cdot \tau_3 \cdot N^{-2-\delta'})$ then $T_q(n,m) = \Omega(n \cdot m^{1-\beta-\delta})$.

Finally, to prove that $m = \Omega(n^{\gamma-\varepsilon}) \cap O(n^{\gamma+\varepsilon})$ notice that:
\begin{itemize}
    \item $m=\Theta(N^{\floor{\rho_1 k}+\floor{\rho_2 k}})$, thus $m\in \Omega(N^{\rho_1 k+\rho_2 k-2}) \cap O(N^{\rho_1 k+\rho_2 k}) = \Omega(N^{\frac{\gamma}{(c+2)} k-2}) \cap O(N^{\frac{\gamma}{(c+2)} k})$.
    \item $n=\Theta(N^{\floor{\rho_3 k}})$, thus $n\in \Omega(N^{\rho_3 k-1}) \cap O(N^{\rho_3 k}) = \Omega(N^{\frac{1}{(c+2)} k-1}) \cap O(N^{\frac{1}{(c+2)} k})$.
\end{itemize}
For sufficiently large $k$ we get $m\in O(N^{\frac{\gamma}{(c+2)} k}) \le O(N^{\frac{\gamma}{(c+2)}k - \gamma + \frac{\varepsilon}{(c+2)}k - \varepsilon}) = O(N^{(\frac{1}{c+2}k-1)(\gamma+\varepsilon)}) \le O(n^{\gamma+\varepsilon})$.

Similarly, $m\in \Omega(N^{\frac{\gamma}{(c+2)} k - 2}) \ge \Omega(N^{\frac{\gamma}{(c+2)}k - \frac{\varepsilon}{(c+2)}k}) = \Omega(N^{\frac{1}{c+2}k(\gamma-\varepsilon)}) \ge \Omega(n^{\gamma-\varepsilon})$.
\end{proof}

\bibliographystyle{alphaurl}
\bibliography{refs}

\end{document}